\documentclass[a4paper]{article}
\usepackage{RR}
\usepackage{hyperref}
\usepackage{cite}
\usepackage{graphicx}
\usepackage{amsmath}
\usepackage[tight,footnotesize]{subfigure}
\newtheorem{theorem}{Theorem}
\newenvironment{proof}[1][Proof]{\begin{trivlist}
\item[\hskip \labelsep {\bfseries #1}]}{\end{trivlist}}
\RRdate{July, 25, 2008} 

\RRauthor{
Ruifeng Zhang 
  \thanks[sfn]{Universit\'e de Lyon, INRIA, INSA-Lyon, CITI, F-69621, FRANCE}%
  \and
Jean-Marie Gorce
\thanksref{sfn}
\and Katia Jaffr\`{e}s-Runser\thanks{Dept. of Electrical and Computer Engineering, Stevens Institute of Technology, Hoboken, New-Jersey 07030, USA}
\thanksref{sfn}
} 
\authorhead{Zhang \& Gorce \& Jaffr\`{e}s-Runser}
\RRtitle{Analyse du compromis \'energie-d\'elai dans les r\'eseaux radio multi-sauts avec liens radio r\'ealistes}
\RRetitle{Energy-delay bounds analysis in wireless multi-hop networks with unreliable radio links} 
\titlehead{Energy-delay Performance Analyses}
\RRnote{This work was supported by INSA-Lyon (PPF PRECIS) and by the Marie Curie program from the European Community's Sixth Framework Program. This chapter only reflects the Author's views and the European Community is not liable for any use that may be made of the information contained herein.}
\RRresume{L'efficacit\'e \'energ\'etique et le d\'elai de transmission de bout en bout sont des param\`etres tr\`es importants pour les r\'eseaux sans fil multi-sauts. Plusieurs \'etudes ont \'et\'e r\'ealis\'ees, sous l'hypoth\`ese de liens radio parfaits (on/off). Pourtant, en environnement r\'eel, les liens radio sont par essence erratiques, et les erreurs de transmission ne peuvent être omises. Dans ce rapport, nous \'etudions le compromis entre \'energie et d\'elai de transmission de bout en bout, pour des communications multi-sauts. Nous d\'erivons sous forme analytique, une borne inf\'erieure du compromis \'energie-d\'elai pour diff\'erents mod\`eles de canaux r\'ealistes (AWGN, Rayleigh, Nakagami), et prenant en compte les processus de retransmission pour garantir l'acheminement des paquets. Cette borne calcul\'ee pour un r\'eseau lin\'eaire, constitue \'egalement une borne inf\'erieure pour les r\'eseaux 2D al\'eatoires. On montre exp\'erimentalement que cette borne est presque atteinte, si la densit\'e de noeuds est suffisante. La principale contribution de ce travail est  l'\'etablissement d'une borne inf\'erieure pour le compromis \'energie-d\'elai, sur canal r\'ealiste et sous contrainte de transmission parfaite de bout en bout. } 
\RRabstract{Energy efficiency and transmission delay are very important parameters for wireless multi-hop networks.
Previous works that study energy efficiency and delay are  based on the assumption of reliable links. However, the unreliability of the channel is inevitable in wireless multi-hop networks. This paper
investigates the trade-off between the energy consumption and the
end-to-end delay of multi-hop communications in a wireless network
using an unreliable link model. It provides a closed form expression
of the lower bound on the energy-delay trade-off for different
channel models (AWGN, Raleigh flat fading and Nakagami block-fading)
in a linear network. These analytical results are also verified in
2-dimensional Poisson networks using simulations. The main
contribution of this work is the use of a probabilistic link model
to define the energy efficiency of the system and capture the
energy-delay trade-offs. Hence, it provides a more realistic lower
bound on both the energy efficiency and the energy-delay trade-off
since it does not restrict the study to the set of perfect links as
proposed in earlier works.} 
\RRmotcle{r\'eseaux de capteurs sans fil, r\'eseaux multi-sauts, compromis \'energie-d\'elai, canal radio r\'ealiste, canal \'a \'evanouissement} \RRkeyword{wireless sensor networks, multi-hop networks, energy-delay trade-off, unreliable links, realistic radio channel, fading channel } 
\RRprojets{ARES} 
\RRtheme{\THCom} 
\URRhoneAlpes 
\begin{document}
\RRNo{6598}
\makeRR   
\section{Introduction}
Energy is a scarce resource for nodes in multi-hop networks such as
Wireless Sensor Networks (WSNs) and Ad-Hoc
networks~\cite{survey:Goldsmith02}. Therefore, energy efficiency is
of paramount importance in most of their applications.

Regarding energy efficiency, there are numerous original works
addressing the problem at the routing layer, MAC layer, physical
layer or from a cross-layer point of view,
e.g.~\cite{routing:HaenggiM05,energy:Chen02,energy:Gao02,energy:Deng07,energy:CuiSG07}.
Routing strategies in multi-hop environments have a major impact on
the energy consumption of networks. Long-hop routes demand
substantial transmission power but minimize the energy cost for
reception, computation and etc. On the opposite, routes made of
shorter hops use fewer transmission power but maximize the energy
cost for reception since there is an increase in the number of hops.
M. Haenggi points out several advantages of using long-hop routing
in his articles, e.g.~\cite{con:Haenggi.M05,routing:HaenggiM05},
among which high energy efficiency is one of the most important
factors. These works reveal the importance of the transmission range
and its impact on the energy conservation but don't provide a
theoretical analysis on the optimal hop length regarding various
networking scenarios. In~\cite{energy:Chen02}, P. Chen et al. define
the optimal one-hop length for multi-hop communications that
minimizes the total energy consumption. They also analyze the
influence of channel parameters on this optimal transmission range.
The same issue is studied in~\cite{energy:Gao02} with a
\textit{Bit-Meter-per-Joule} metric where the authors study the
effects of the network topology, the node density and the
transceiver characteristics on the overall energy expenditure. This
work is improved by J.~Deng et al. in~\cite{energy:Deng07}.

Since the data transmitted is often of a timely nature, the
end-to-end transmission delay becomes a an important performance
metric. Hence, minimum energy paths and the trade-off between energy
and delay have been widely studied,
e.g.~\cite{energy:CuiSG07,energy:Cui05,energy:Chang05}. However,
unreliable links are not considered in the aforementioned works. In
fact, experiments in different environments and theoretical analyzes
in~\cite{con:Ganesan03,con:Woo03,cor:Zhao03,model:Zamalloa07,con:Gorce07}
have proved that unreliable links have a strong impact on the
performance of upper layers such as MAC and routing layer. In our
previous work~\cite{con:Gorce07}, we have shown how unreliable link
improve the connectivity of WSNs.


In~\cite{energy:Banerjee04}, S. Banerjee et al. take unreliable
links into account in their energy efficiency analysis by
introducing a link probability and the effect of link error rate.
The authors derive the minimum energy paths for a given pair of
source and destination nodes and propose the corresponding routing
algorithm. However, the energy model used in this paper includes the
transmission power only and does not consider circuitry energy
consumption at the transmitter and receiver side. In fact, such a
model leads to an unrealistic conclusion which states that the
smaller hop distance, the higher energy efficiency. As we show in
this paper, considering a constant circuitry power according
to~\cite{com:Karl05} results in completely different conclusions.
Furthermore, we propose to evaluate the effect of fading channel on
the energy efficiency.

In this work, we do not consider any specific protocol and assume
the corresponding overhead to be negligible. Depending on the
application, the energy efficiency has a different
significance~\cite{com:Karl05}. A periodic monitoring application is
assumed here where the energy spent per correctly received bit is a
crucial energy metric. Moreover, in wireless communications, the
energy cost augments with the increase of the transmission distance.
Hence, we also adopt the mean Energy Distance Ratio per bit
($\overline{EDRb}$) metric in $J/m/bit$ proposed in
\cite{energy:Gao02}. A realistic unreliable link
model~\cite{con:Gorce07} is introduced into the energy model. The
purpose of this work is to provide a lower bound on the energy
efficiency of both single and multi-hop transmissions and derive the
corresponding average transmission delay. As such, we are able to
show the theoretical trade-off between the energy efficiency and the
delay for single-hop and multi-hop transmissions. The multi-hop case
is analyzed in a homogeneous linear network. Both studies are
performed over three different channels (i.e. AWGN, Rayleigh flat
fading and Nakagami block fading channel). Theoretical results are
then validated in 2-dimensional Poisson distributed network using
simulations.

The contributions of this paper are:
\begin{itemize}
\item The close-form expressions for the optimal transmission range and for the corresponding  optimal transmission power are derived in AWGN channel, Rayleigh flat fading channel and Nakagami block fading channel employing both a comprehensive energy model and an unreliable link
model.
\item The definition of a closed form expression for the lower bound of energy efficiency of a multi-hop communication is obtained in a linear network over three types of channel and is validated by simulation in 2-dimensional Poisson networks.
\item The definition of a lower bound for the energy-delay trade-off for a linear and a Poisson network in the three types of channel
aforementioned.
\end{itemize}

This paper is organized as follows:
Section~\ref{sec:model_metric} concentrates on presenting the models
and metrics used in the paper. Section~\ref{sec:one hop} derives a
closed form expression of optimal transmission range and optimal
transmission power for one-hop transmission. In
section~\ref{sec:low_bound}, the minimum energy reliable path for
linear networks and its delay are deduced. In
section~\ref{sec:en_delay}, we focus on the optimal trade-off
between the energy consumption and the delay in linear networks.
Simulations are given and analyzed in section~\ref{sec:sim} for a
2-dimensional network. Finally, section~\ref{sec:conclu} concludes
our work.



\section{Models and Metric}
\label{sec:model_metric} In this section, the energy model, the
realistic unreliable link model, the delay model and the metric
$\overline{EDRb}$ used in this work are introduced.
\subsection{Energy consumption model}
\label{subsec:model_energy} We consider energy efficient nodes, i.e.
nodes that only listen to the transmissions intended to themselves
and that send an acknowledgment packet (ACK) to the source node
after a correct packet reception. As such, the energy consumption
for transmission of one packet $E_p$ is composed of three
parts\footnote{In this works, no coding is considered, so the energy
cost for coding/decoding is set to zero.}: the energy consumed by
the transmitter $E_{Tx}$, by the receiver $E_{Rx}$ and by the
acknowledgement packet exchange $E_{ACK}$:
\begin{equation}
\label{eq:Ep}
 E_{p} = E_{Tx}+E_{Rx}+E_{ACK}  
\end{equation}

The transmission energy model~\cite{com:Karl05} is given by:
\begin{equation}
\label{eq:Etx} E_{Tx} = T_{start}\cdot P_{start}+\frac{N_b}{R}\cdot
(P_{txElec}+\alpha_{amp} +\beta_{amp}\cdot P_t)
\end{equation}
where $P_t$ is transmission power, the other parameters are
described in Table~\ref{tab:para} and $P_{txElec}$ is considered as
constant.


\begin{table}[!t]
\centering
    \caption{Some parameters of the transceiver energy consumption\cite{com:Karl05}}
    \label{tab:para}
    \begin{center}
        \begin{tabular}{|l|l|l|}
            \hline
            Symbol& Description & Value  \\
            \hline
            $P_{start}$& Startup power&$58.7~mW$  \\
            $T_{start}$& Startup time&$446~\mu s$  \\
            $P_{txElec}$& Transmitter circuitry power &$151~mW$ \\
            $\alpha_{amp}$& Amplifier constant power &$174~mW$ \\
            $\beta_{amp}$& Amplifier proportional offset ($>1$) & $5.0$ \\
            $P_{rxElec}$& Receiver circuitry power & $279~mW$\\
            $N_b$ & Number of bits per packet & $2560$\\
            $R $& Transmission bit rate &$ 1~Mbps $\\
            $N_0$& Noise level&$-154dBm/Hz$ \\
            $f_c$ & Carrier frequency &$2.4GHz$ \\
            $G_{Tant}$ & Transmitter antenna gain & $1$\\
            $G_{Rant}$ & Receiver antenna gain & $1$\\
            $\alpha$ & Path-loss exponent & $3$\\
            $L$ & &$1$\\
            $T_{ACK}$ & ACK Duration &$5mS$\\
            \hline
        \end{tabular}
    \end{center}
\end{table}

Similarly, the energy model on the receiver side includes two parts:
the startup energy consumption, which is considered identical to the
one of the transmitter, and the circuitry cost~\cite{com:Karl05}:
\begin{align}
\label{eq:Erx}
 E_{Rx} = T_{start}\cdot P_{start}+\frac{N_b}{R}\cdot P_{rxElec}. 
\end{align}
where $P_{rxElec}$ is the circuity power of the receiver which is
considered as constant.

In the acknowledgment process, it is assumed that the ACK packet can
be successfully transmitted in a single attempt which is based on
the following facts: firstly, since ACK packets are much smaller
data packets, their link probability is greater than that of data
packet. For instance, for respectively ACK and Data packets of 80
and 320 bytes each, if the successful transmission probability of
the data packet is  $80\%$, the link probability for the ACK packet
is  $95\%$. Secondly, assuming a symmetric channel, if the data
packet experienced a good channel, the return path experiences the
same beneficial channel conditions. Hence, we can assume that only
one ACK packet is sent with high probability of success to the
source of the message.

Since the energy consumed by the transmission power $P_t$ for the
ACK packet has small proportion in its total energy, $P_t$ is
neglected in the energy expenditure model given by:
\begin{equation}
\label{eq:Eack} E_{ACK} = (P_{txElec}+P_{rxElec}+\alpha_{amp})\cdot
T_{ACK},
\end{equation}
where $T_{ACK}$ is the average time during which the transmitter
waits for an ACK packet.

The analysis of $E_p$ shows that the energy consumption can be
classified into two parts: the first part is constant, including
$T_{start}\cdot P_{start}$, $P_{txElec}$, $~\alpha_{amp}$,
$P_{rxElec}$ and $E_{ACK}$, which are independent of the
transmission range; the second part is variable and depends on the
transmission energy $P_t$ which is tightly related to the
transmission range. Accordingly, the energy model for each bit
follows:
\begin{equation}
\label{eq:Eb} E_{b} = \frac{E_p}{N_b}= E_{c}+K_1\cdot P_{t},
\end{equation}
where $E_b$, $E_c$ and $K_1\cdot P_{t}$ are respectively the total,
the constant and the variable energy consumption per bit.
Substituting \eqref{eq:Ep} \eqref{eq:Etx} \eqref{eq:Erx}
\eqref{eq:Eack} into \eqref{eq:Eb} yields:
\begin{align}
E_c = &\frac{2T_{start}\cdot P_{start}}{N_b}+ \notag\\
&(P_{txElec}+P_{rxElec}+\alpha_{amp})(\frac{1}{R}+\frac{T_{ACK}}{N_b}) \label{eq:Ec} \\
K_1 = &\frac{\beta_{amp}}{R} \label{eq:K1}.
\end{align}
For a given transmitting/receiving technology, $E_c$ and $K_1$ are
constant because all parameters in~\eqref{eq:Ec} and \eqref{eq:K1}
are fixed. Then $E_b$ becomes a function of $P_t$, i.e., $E_b(P_t)$.

\subsection{Realistic unreliable link model}
\label{subsec:model_link}  The unreliable radio link model is
defined using the packet error rate (PER)~\cite{con:Gorce07}:
\begin{equation}
\label{eq:pl} p_l(\gamma_{x,x'})= 1-PER(\gamma_{x,x'})
\end{equation}
where $PER(\gamma)$ is the PER obtained for a signal to noise ratio
(SNR) of $\gamma$. The PER depends on the transmission chain
technology (modulation, coding, diversity ... ). And $\gamma_{x,x'}$
is calculated by~\cite{com:Karl05}:
\begin{equation}
\label{eq:gamma} \gamma_{x,x'} = K_2 \cdot P_t \cdot
d_{x,x'}^{-\alpha},
\end{equation}
with
\begin{equation}
K_2 = \frac{G_{Tant}\cdot G_{Rant}\cdot \lambda^2}{(4\pi)^2 N_0\cdot
B \cdot L},
\end{equation}
where $d_{x,x'}$ is the transmission distance between node $x$ and
$x'$, $\alpha\geq2$ is the path loss exponent, $P_t$ is the
transmission power, $G_{Tant}$ and $G_{Rant}$ are the antenna gains
for the transmitter and receiver respectively, $B$ is the bandwidth
of the channel and is set to $B = R$, $\lambda$ is the wavelength
and $L\geq1$ summarizes losses through the transmitter and receiver
circuitry.

Similar to $E_b$, for a given technology, $K_2$ becomes a constant.
And $p_l(\gamma(x,x'))$ can be rewritten as a function of $d$ and
$P_t$, i.e., $p_l(d,P_t)$.

\subsection{Mean energy distance ratio per bit ($\overline{EDRb}$)}

The mean Energy Distance Ratio per bit
($\overline{EDRb}$)~\cite{energy:Gao02} in $J/bit/m$ is defined as
the energy consumption for transmitting one bit over one meter. The
mean energy consumption per bit for the successful transmission over
one hop  $\overline E_{1hop}$ including the energy needed for
retransmissions is given by:
\begin{align}
\label{eq:E1hop} \overline E_{1hop}&=E_b(P_t)\cdot
\sum_{n=1}^{\infty}n\cdot p_l(d,P_t)\cdot
(1-p_l(d,P_t))^{(n-1)}\notag
\\ &=E_b(P_t)\cdot \frac{1}{p_l(d,P_t)}
\end{align}
where $n$ is the number of retransmissions. According to its
definition, $\overline{EDRb}$ is given by:
\begin{equation}
\label{eq:EDRb}
 \overline{EDRb} = \frac{\overline E_{1hop}}{d}
  =\frac{E_b(P_t)}{d\cdot p_l(d,P_t)}=\frac{E_{c}+K_1\cdot
  P_{t}}{d\cdot p_l(d,P_t)}.
\end{equation}

\subsection{Delay model}
The average delay for a packet to be transmitted over one hop,
$D_{onehop}$, is defined as the sum of three delay components. The
first component is the queuing delay during which a packet waits for
being transmitted. The second component is the transmission delay
that is equal to $N_b/R$. The third component is $T_{ACK}$. Note
that we neglect the propagation delay because the transmission
distance between two nodes is usually short in multi-hop networks.
Without loss of generality, $D_{onehop}$ is set to be $1$ unit.
However, one-hop transmission may suffer from the delay caused by
retransmissions. According to~\eqref{eq:E1hop}, the mean delay of a
reliable one-hop transmission is:
\begin{align}
\label{eq:delay_onehop} \overline D &= D_{onehop}\times mean\
number\ of\ retransmissions \notag \\ &=\frac{1}{p_l(d,P_t)}.
\end{align}

\section{One-hop Transmission: Energy Efficiency and Delay}
\label{sec:one hop} The one-hop transmission is the building block
of a multi-hop path. In this section, we derive the optimal
transmission range and power that minimizes the energy expenditure
of the one-hop transmission by introducing three different channel
models. Optimal transmission range $d_0$ and optimal transmission
power $P_0$ are calculated according to:
\begin{align}
&\frac{\partial{\overline{EDRb}} }{\partial{P_{t}}}
=\frac{\partial}{\partial{P_{t}}}\left( \frac{E_b(P_{t})}{d\cdot
p_l(d,P_{t})}\right)\Bigg|
_{P_{t}=P_0}&=0 \label{eq:dPt_EDRb} \\
&\frac{\partial{\overline{EDRb}} }{\partial{d}}
=\frac{\partial}{\partial d}\left(\frac{E_b(P_{t})}{d\cdot
p_l(d,P_{t})}\right)\Bigg| _{d=d_0}&=0. \label{eq:dd_EDRb}
\end{align}


The optimal transmission power $P_0$ and range $d_0$ exist because
for smaller values of $d$, the transmission power $P_t$ is low in
terms of a certain link probability and the constant energy
component $E_c$ is dominating in $\overline{EDRb}$ consequently; for
higher values of $d$, the variable energy consumption $K_1\cdot P_t$
is dominating $\overline{EDRb}$ since $P_t$ increases proportionally
to $d^\alpha$ in order to reach the destination.

\subsection{Energy-optimal transmission power $P_0$}
Substituting \eqref{eq:pl} into \eqref{eq:dPt_EDRb} and
\eqref{eq:dd_EDRb} and simplifying, according to the derivation in
Appendix~\ref{sec:appendix}, we obtain:
\begin{equation}\label{eq:P0} P_0 = \frac{E_c}{K_1(\alpha -1)}.
\end{equation}
Substituting~\eqref{eq:Ec} and~\eqref{eq:K1} into~\eqref{eq:P0}
yields:
\begin{align}\label{eq:P0_close}P_0=&\frac{1}{\alpha - 1}\left(\frac{2T_{start}\cdot
P_{start}}{\frac{N_b}{R} \beta_{amp}}\right) \notag
\\&+\left(\frac{1}{\beta_{amp}}+  \frac{T_{ACK}}{\frac{N_b}{R} \beta_{amp}} \right)
\frac{P_{txElec}+P_{rx}+\alpha_{amp}}{\alpha-1},
\end{align}
where $N_b/R$ is the transmission duration.

In~\eqref{eq:P0_close}, it should be noted that $P_0$ is independent
from $p_l(\gamma)$ and consequently independent from modulation and
fading. In general, $N_b/R\gg T_{start}$. Following, the first part
of \eqref{eq:P0_close} can be neglected. On the opposite, the
characteristics of the amplifier have a strong impact on $P_0$. When
the efficiency of the amplifier is high, i.e.
$\beta_{amp}\rightarrow 1$, $P_0$ reaches its maximum value
resulting in a longer optimal transmission range $d_0$. It tallies
with the result of~\cite{routing:HaenggiM03}. It is clear that when
the environment of transmission deteriorates, namely, $\alpha$
increases, $P_0$ decreases correspondingly.

\subsection{Energy-optimal transmission range $d_0$ and its delay}

According to Appendix~\ref{sec:appendix}, $d_0$ follows:
\begin{equation}\label{eq:d0_general}
d_0 ^\alpha = \frac{p_l^\prime (\gamma) K_2 P_0}{p_l(\gamma) },~{\rm
with~}  \gamma = K_2 P_0 d_0 ^{-\alpha}
\end{equation}
\noindent where $p_l^\prime(\gamma)$ is the first derivative of
$p_l(\gamma)$. Equation \eqref{eq:d0_general} indicates that $d_0$
depends on $p_l(\gamma)$ and hence has to be analyzed according to
the type of channel and modulation as proposed next. This expression
is meaningful since it can be used to estimate the optimal node
density in a wireless netework depending on $pl(\gamma)$.

\subsubsection{AWGN channel}
The optimal transmission range in AWGN channel, which is derived in
Appendix~\ref{sec:d0g_der_awgn}, is obtained by:
\begin{equation}
\label{eq:d0_awgn} d_{0g} = \left(\frac{-0.5415\beta_m K_2 N_b E_c
\alpha}{K_1(\alpha -1)(1+\alpha N_b \text
W_{-1}\left[\frac{-e^{-\frac{1}{N_b\cdot\alpha}}}{0.1826\alpha_m N_b
\alpha}\right])}\right)^{\frac{1}{\alpha}}
\end{equation}
where $\text W_{-1}[\cdot]$ is the branch satisfying $\text W(x)<-1$
of the Lambert W function~\cite{math:Corless96}.

Substituting~\eqref{eq:P0} and~\eqref{eq:d0_awgn}
into~\eqref{eq:gamma}, the optimal SNR $\gamma_{0g}$ is given by:
\begin{equation}
\label{eq:gamma0_awgn} \gamma_{0g}=\frac{1+\alpha N_b \text
W_{-1}\left[-\frac{e^{-\frac{1}{N_b\cdot\alpha}}}{0.1826\alpha_m N_b
\alpha}\right]}{-0.5415\beta_m k N_b \alpha}.
\end{equation}
Meanwhile, the optimal BER  is obtained by \eqref{eq:gamma0_awgn}
and \eqref{eq:ber_awgn}:
\begin{equation}
\label{eq:BER0_g} BER_{0g} =0.1826\alpha_m \exp(\frac{1+\alpha N_b
\text
W_{-1}\left[\frac{-e^{-\frac{1}{N_b\cdot\alpha}}}{0.1826\alpha_m N_b
\alpha}\right]}{N_b \alpha}).
\end{equation}
Depending on $\gamma_{0g}$ and $BER_{0g}$, the receiver can decide
whether it is in the optimal communication range or not by measuring
its channel state.

The delay and the energy efficiency of the one-hop communication can
be analyzed by expressing respectively the delay $D_g$ and the
energy metric $\overline{EDRb}$ as a function of the transmission
range $d$ as detailed in Appendix~\ref{sec:d0g_der_awgn}. Hence,
substituting~\eqref{eq:ber_awgn} into~\eqref{eq:delay_onehop}, the
delay of the reliable one-hop transmission in AWGN channel as a
function of $d$ is given by:
\begin{equation}
\label{eq:delay_oh_gaussian} \overline D_g =
\left(1-0.1826\alpha_m\cdot\exp(-0.5415\beta_mK_2 P_t
d^{-\alpha})\right)^{-N_b}.
\end{equation}

Substituting~\eqref{eq:pl_gaussian} into~\eqref{eq:dPt_EDRb}, the
optimal transmission power $P_{0g}$ as a function of the
transmission distance $d$ achieving energy efficiency in AWGN
channel follows:
\begin{equation}
\label{eq:P0_d_g} P_{0g}(d) = \frac{d^\alpha+N_b d^\alpha \text
W_{-1}\left[\frac{\exp\left(-\frac{0.5415 E_c K_2 \beta_m }{K_1
d^\alpha}-\frac{1}{N_b}\right)}{-0.1826\alpha_m
N_b}\right]}{-0.5415\beta_m K_2 N_b}-\frac{E_c}{K_1}.
\end{equation}
Substituting~\eqref{eq:P0_d_g} and~\eqref{eq:ber_awgn}
into~\eqref{eq:EDRb}, $\overline{EDRb}$ as a function of $d$ in AWGN
channel is expressed by:
\begin{equation}
\label{eq:EDRb_d_g} \overline{EDRb}(d)= \frac{E_c +
K_1P_{0g}(d)}{d\left(1-BER_g(K_2 P_{0g}(d)
d^{-\alpha})\right)^{N_b}}.
\end{equation}

\begin{figure}[!t]
\centering
\includegraphics[width=0.8\textwidth]{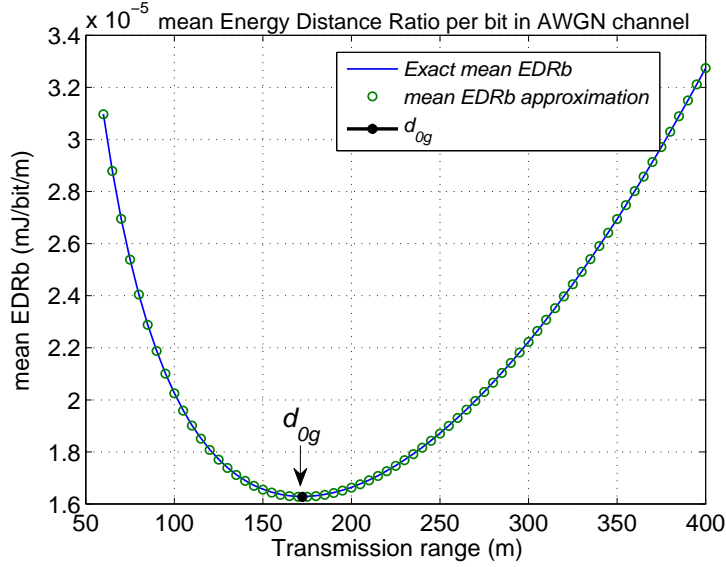}
\caption{$\overline{EDRb}$ of the one-hop transmission as a function
of the range $d$ in AWGN channel where $d_{0g}=172.31m$, $ P0 =
180.51 mW$, $\gamma_{0g}=9.34dB$, $BER_{0g}=1.37e-5$, $ p_l =
96.55\%$ and $\overline D_g = 1.04$ unit. The exact
$\overline{EDRb}$ is obtained with~\eqref{eq:ber_awgn_general} and
the approximation of $\overline{EDRb}$ is obtained with
\eqref{eq:ber_awgn}. Therefore, the approximation is feasible.
\label{fig:EDRb_gaussian}.}
\end{figure}

Fig.~\ref{fig:EDRb_gaussian} shows the variation of
$\overline{EDRb}$ with the transmission range $d$ in AWGN channel as
an example according to~\eqref{eq:EDRb_d_g}, where BPSK modulation
is adopted. The related parameters are listed in
Table~\ref{tab:para}.
It should be noted that the value of $p_l$ is close to $1$, which
shows that energy optimal links in AWGN channel are reliable.

\subsubsection{Rayleigh flat fading channel~\cite{energy:ZhangRF08}}
The optimal transmission range $d_{0f}$ in Rayleigh flat fading
channel, which is derived in Appendix~\ref{sec:d0f_der_fading}, is
obtained by:
\begin{equation}
\label{eq:d0_fading} d_{0f}=\left(\frac{2\beta_m
E_c\cdot{K_2}}{(\alpha -1) \cdot{K_1}\alpha_m (\alpha
N_b+1)}\right)^{\frac{1}{\alpha}}.
\end{equation}
The expression of $d_{0f}$ shows that it decreases with the increase
of $\alpha$ or $N_b$.

Substituting~\eqref{eq:P0} and \eqref{eq:d0_fading}
into~\eqref{eq:gamma} provides the optimal SNR in Rayleigh flat
fading channel:
\begin{equation}
\label{eq:gamma0}\bar\gamma_{0f} = (\alpha N_b
+1)\frac{\alpha_m}{2\beta_m}\approx\frac{\alpha
N_b\cdot\alpha_m}{2\beta_m} ~~ (\alpha N_b\gg1).
\end{equation}
Then substituting ~\eqref{eq:gamma0} into \eqref{eq:berf}, the
optimal BER in Rayleigh flat fading channel is:
\begin{equation}
\label{eq:BER0_f} BER_{0f} = \frac{1}{\alpha N_b +1} \approx
\frac{1}{\alpha N_b}.
\end{equation}
From a cross layer point of view, the routing layer can identify if
a node is at the optimal communication range according to the values
of $\bar \gamma_{0f}$ or $BER_{0f}$.

Similarly to the study in AWGN channel, we derive here the
expression of the delay $D_f$ and the energy metric
$\overline{EDRb}$ as a function of the transmission range $d$ which
is detailed in Appendix~\ref{sec:d0f_der_fading}.
Substituting~\eqref{eq:berf} into~\eqref{eq:delay_onehop}, the delay
of the reliable one-hop transmission in a Rayleigh flat fading
channel as a function of $d$ is given by:
\begin{equation}
\label{eq:delay_oh_fading} \overline D_f =
\left(1-\frac{\alpha_m}{2\beta_m K_2 P_t d^{-\alpha}}\right)^{-N_b}.
\end{equation}

Substituting~\eqref{eq:EDRb_fading} into~\eqref{eq:dPt_EDRb}, the
optimal transmission power $P_{0f}$ as a function of the
transmission distance $d$ achieving energy efficiency in a Rayleigh
channel follows:
\begin{align}
\label{eq:P0_d} P_{0f}(d)=&\frac{d^a(1+N_b)\alpha_m }{4K_2\beta_m}
+\frac{\sqrt{d^\alpha K_1\alpha_m(d^\alpha
K_1(1+N_b)^2\alpha_m+8E_cK_2N_b\beta_m)}}{4K_2 K_1\beta_m}
\end{align}
Hence, for a given transmission distance, the optimal transmission
power can be derived according to $P_{0f}(d)$ in an adaptive power
configuration.

Finally, $\overline{EDRb}$ as a function of $d$ is computed by
substituting ~\eqref{eq:P0_d} into \eqref{eq:EDRb}:
\begin{equation}
\label{eq:EDRb_d_f} \overline{EDRb}(d)= \frac{E_c +
K_1P_{0f}(d)}{d\left(1-\frac{\alpha_m}{2\beta_mK_2\cdot P_{0f}(d)
\cdot d^{-\alpha}}\right)^{N_b}}
\end{equation}

\begin{figure}[!t]
\centering
\includegraphics[width=0.8\textwidth]{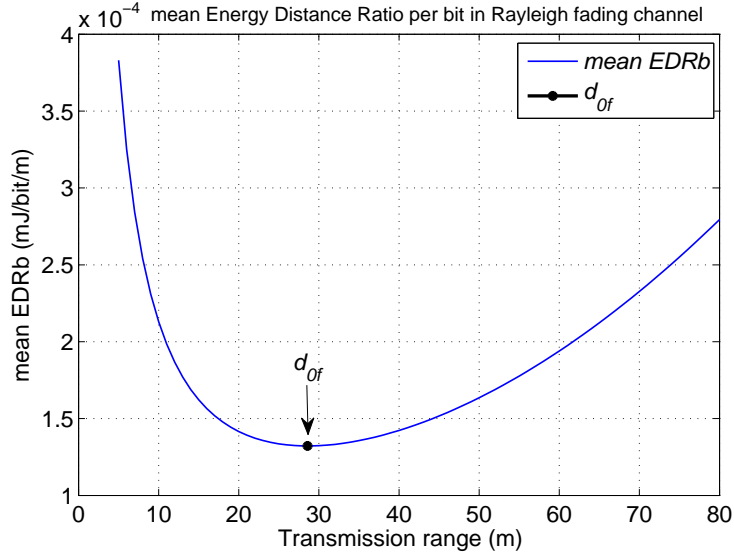}
\caption{$\overline{EDRb}$ of the one-hop transmission as a function
of the range $d$ in Rayleigh flat fading channel where
$P_0=180.51mW$, $d_{0f}=15.97m$, $\gamma_{0f}=32.83dB$,
$BER_{0f}=1.30e-4$, $p_l = 0.72$ and $\overline D_f= 1.40$.}
\label{fig:EDRb_fading}
\end{figure}
Fig.~\ref{fig:EDRb_fading} shows that $\overline{EDRb}$ varies with
$d$ according to~\eqref{eq:EDRb_d_f} in a Rayleigh flat fading
channel. The parameters related are listed in table~\ref{tab:para}.
Having $p_l = 0.7165$ shows that an energy optimal link in Rayleigh
channel is far less reliable than the link in AWGN channel. This
result claims for using unreliable links in the real deployment of
wireless network.

\subsubsection{Nakagami block fading channel}
The link model in Nakagami block fading channel, as shown
in~\eqref{eq:pl_block}, is too complex to obtain the closed form
expression of the energy optimal transmission distance $d_{0b}$.
Therefore, two scenarios are taken into consideration in the
following.

Firstly, when $m =1$ and $\alpha_m =1$ (e.g., for BPSK, BFSK and
QPSK), according to the derivation in
Appendix~\ref{sec:d0b_der_block}, the optimal transmission range
$d_{0b}$ in Nakagami block fading channel is:
\begin{equation}
d_{0b}=\left(\frac{ \beta_m K_2 E_c
}{K_1(\alpha^2-\alpha)(4.25\log_{10}(N_b)-2.2)} \right) ^{1/\alpha}.
\label{eq:d0_block}
\end{equation}
Substituting~\eqref{eq:P0} and \eqref{eq:d0_block}
into~\eqref{eq:gamma} yields the optimal signal to noise ratio in
Nakagami block fading channel:
\begin{equation}
\overline{\gamma}_{0b}=\frac{\alpha}{\beta_m}(4.25\log_{10}(N_b)-2.2)
\end{equation}

For a given transmission range, we can obtain the optimal
transmission power $P_{0b}$ in Nakagami block fading channel using
~\eqref{eq:dPt_EDRb} and~\eqref{eq:pl_block_approx}:
\begin{equation}
P_{0b}(d)=\frac{-2E_c}{K_1-\frac{\sqrt{K_1(4E_c \beta_m K_2 -2.2
d^\alpha K_1 + 4.25 d^\alpha K_1
\log_{10}(N_b))}}{\sqrt{d^\alpha(4.25\log_{10}(N_b)-2.2)}}}.
\label{eq:P0_d_b}
\end{equation}
Finally, $\overline{EDRb}$ as a function of $d$ is obtained by
substituting \eqref{eq:P0_d_b} into \eqref{eq:EDRb}:
\begin{equation}
\label{eq:EDRb_d_b} \overline{EDRb}(d)= \frac{E_c +
K_1P_{0b}(d)}{d\exp\left(\frac{-4.25\log_{10}(Nb)+2.2}{\beta_m
K_2P_{0b}(d)d^{-\alpha}}\right)^{N_b}}.
\end{equation}
Substituting~\eqref{eq:pl_block} into~\eqref{eq:delay_onehop}, we
get the delay of a reliable one-hop $\overline D_b$ in Nakagami
block fading channel:
\begin{equation}
\label{eq:delay_oh_block} \overline D_b =
\frac{1}{\int_{\gamma=0}^{\infty} (1-BER(\gamma))^{N_b} p(\gamma |
K_2P_t  d^{-\alpha})d\gamma}.
\end{equation}

For the other scenarios, the sequential quadratic programming (SQP)
method algorithm in~\cite{math:Ravindran06} is adopted to solve the
optimization problem related to the computation of the optimal
$\overline{EDRb}$.

\begin{figure}[!t]
\centering
\includegraphics[width=0.8\textwidth]{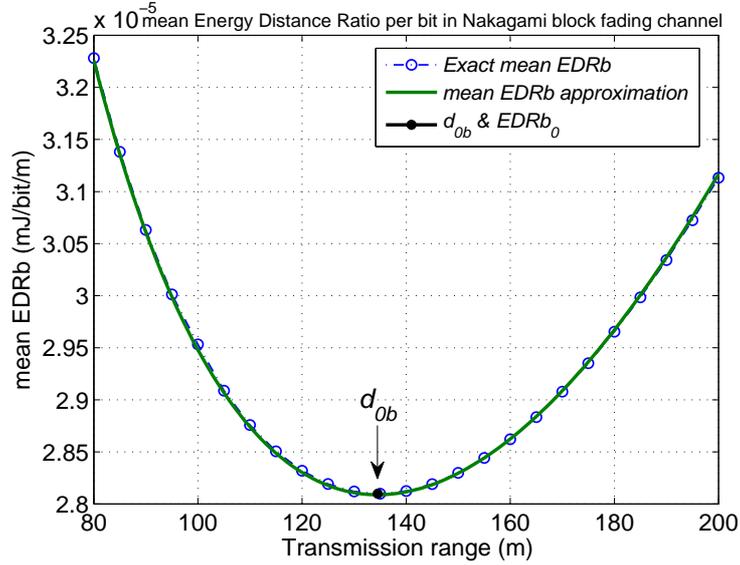}
\caption{$\overline{EDRb}$ of the one-hop transmission as a function
of the range $d$ in Nakagami-m block fading channel where $m=1$,
$\alpha_m=1$, $P_0=180.50mW$, $d_{0b}=134.16m$,
$\gamma_{0b}=12.69dB$, $BER_{0b}=6.67e-005$, $p_l = 0.72$ and
$\overline D_b= 1.40$. The exact $\overline{EDRb}$ is obtained
with~\eqref{eq:pl_block} and the approximation of $\overline{EDRb}$
is obtained with~\eqref{eq:pl_block_approx}. Therefore, this is an
suitable approximation.} \label{fig:EDRb_block}
\end{figure}
Fig.~\ref{fig:EDRb_block} shows how $\overline{EDRb}$ varies with
$d$ according to Eq.~\eqref{eq:EDRb_d_b} in  Nakagami block fading
channel using BPSK modulation. The related parameters are presented
in table~\ref{tab:para}.  Having $p_l = 0.72$ reveals that energy
optimal links in Nakagami block fading channel are even more
unreliable than those in Rayleigh flat fading channel.

From Fig.~\ref{fig:EDRb_gaussian}, Fig.~\ref{fig:EDRb_fading} and
Fig.~\ref{fig:EDRb_block}, it can be concluded that: firstly, the
optimal transmission power $P_0$ corresponding to the optimal
transmission range is the same for all channels which concises with
the result of Eq.\eqref{eq:P0}; secondly, the optimal transmission
range decreases when fading becomes stronger, namely, from AWGN,
Nakagami block fading channel to Rayleigh flat fading channel;
thirdly, $\overline{EDRb}$ increases with the enlargement of fading
, i.e., more energy has to be consumed to counteract the effect of
fading.

%
%
%
%

\subsection{Impact of some physical parameters}
This section studies the impact of some physical parameters such as
the path-loss exponent $\alpha$, the strength of fading, the
circuitry power, $N_b$, the transmission rate $R$ and the modulation
technique. For all the results provided hereafter, the values of
physical parameters that are not analyzed are given in
Table~\ref{tab:para}.

\paragraph*{Impact of fading}
\begin{figure}[!t]
    \centering
    \includegraphics[width=0.8\textwidth]{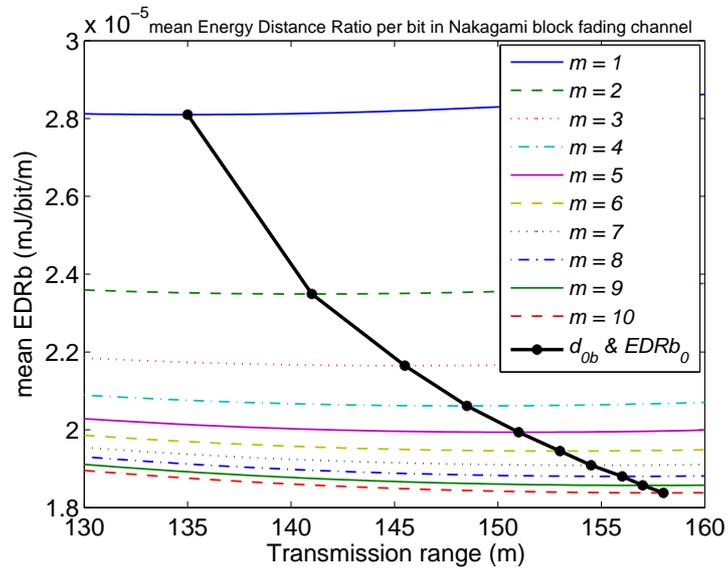}
    \caption{Impact of the strength of fading ($m$) in Nakagami block fading channel
    for the one-hop transmission energy
    performance given by $\overline{EDRb}$ as a function of $d$.}
    \label{fig:effect_m}
\end{figure}
The sequential quadratic programming (SQP) algorithm described
in~\cite{math:Ravindran06} is implemented to analyze the impact of
strength of fading on the optimal $\overline{EDRb}$ and
corresponding optimal transmission range in Nakagami block fading
channel. The results are shown in Fig.~\ref{fig:effect_m}. Similarly
to our previous analysis, the increase of the strength of fading
leads to the increase of the optimal $\overline{EDRb}$ and shortens
the optimal transmission range. In that case, more energy is
consumed to overcome the destructive effect of fading.

\paragraph*{Impact of the path loss exponent}
\begin{figure*}[!t]
    \centering
    \subfigure[AWGN channel]{
    \includegraphics[width=0.6\textwidth]{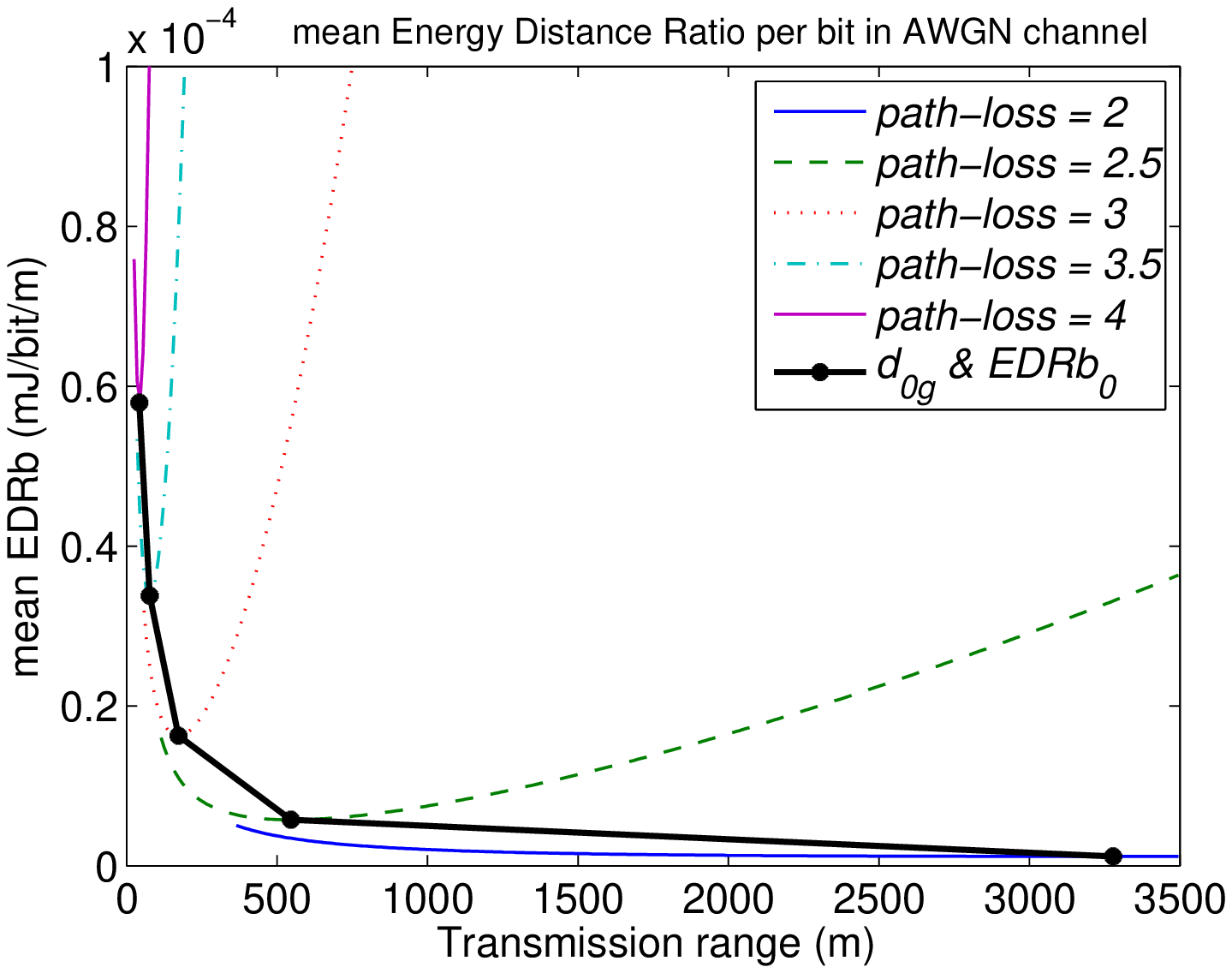}}
    \hfil
    \subfigure[Rayleigh flat fading channel]{
    \includegraphics[width=0.6\textwidth]{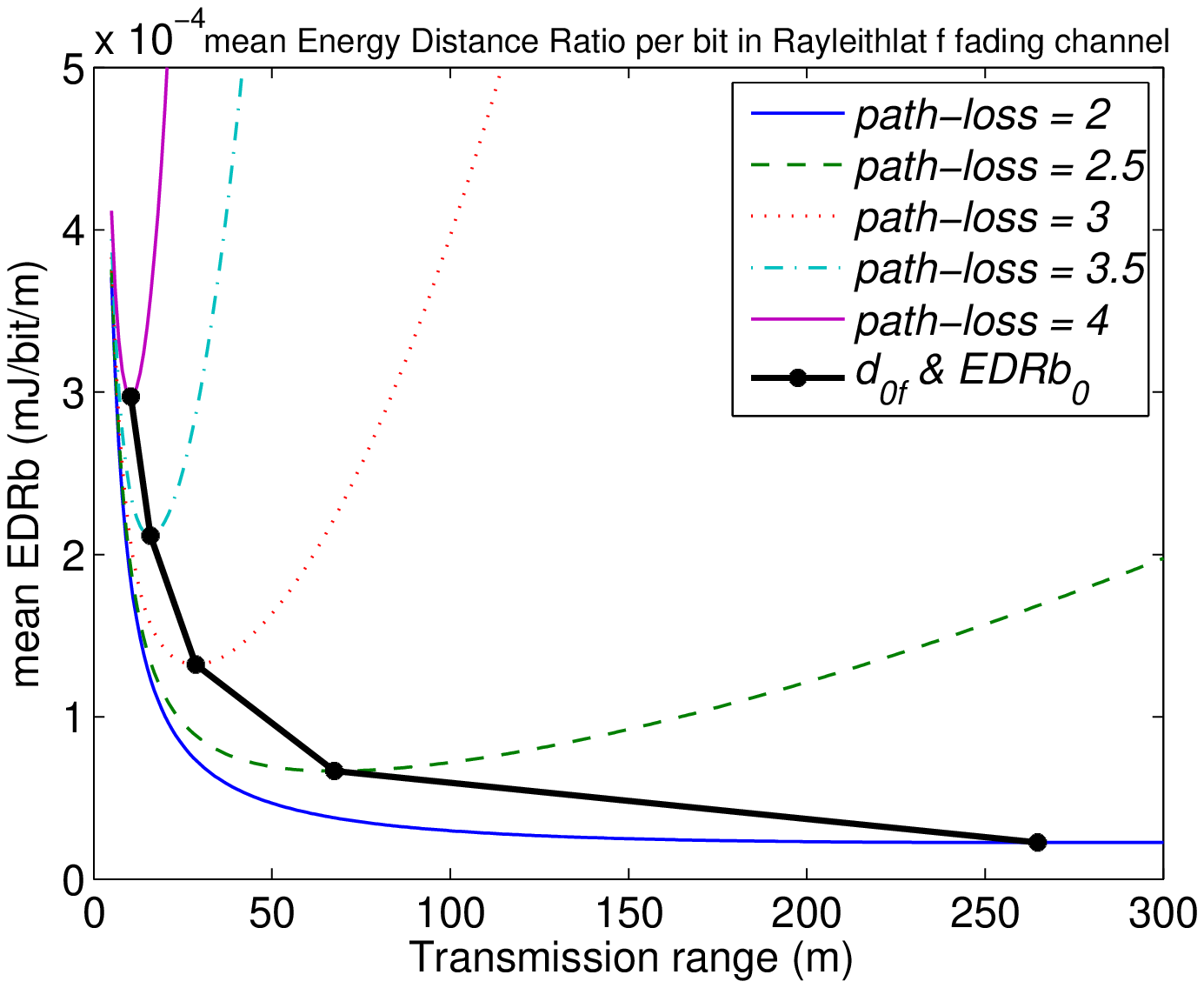}}
    \hfil
    \subfigure[Nakagami block fading channel $m=1$]{
    \includegraphics[width=0.6 \textwidth]{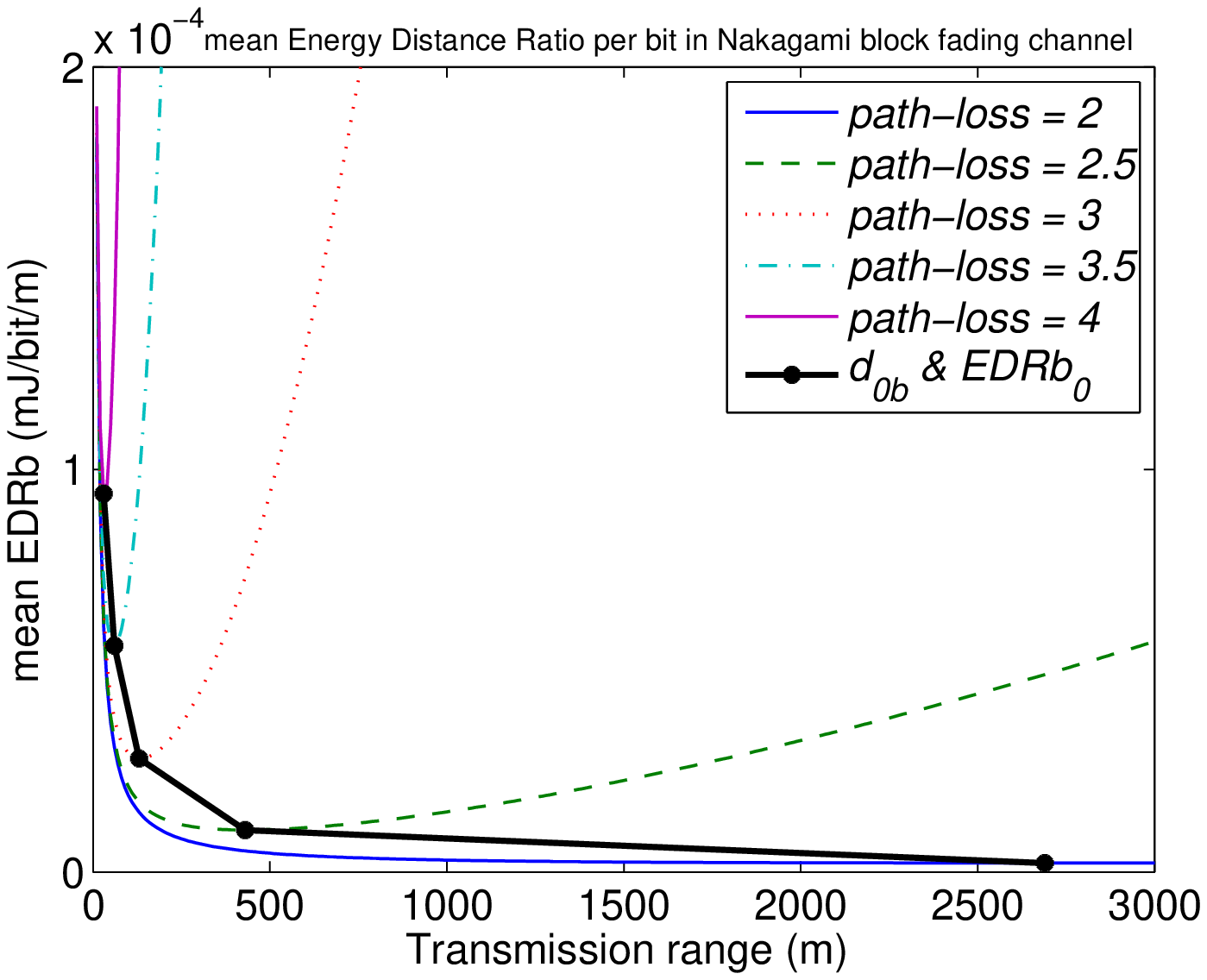}}
    \caption{Impact of the path loss exponent $\alpha$ on the one-hop transmission energy
    performance given by $\overline{EDRb}$ as a function of $d$.}
     \label{fig:effect_pathloss}
 \end{figure*}
Fig.~\ref{fig:effect_pathloss} shows that $\overline{EDRb}$ greatly
increases with the strength of the path loss, i.e., more energy is
consumed to make up for the path loss. Meanwhile, path loss shortens
the optimal transmission range which induces more hops and higher
delay for a given transmission distance.

\paragraph*{Impact of the circuitry power}
\begin{figure*}[!t]
    \centering
    \subfigure[AWGN channel]{
    \includegraphics[width=0.6\textwidth]{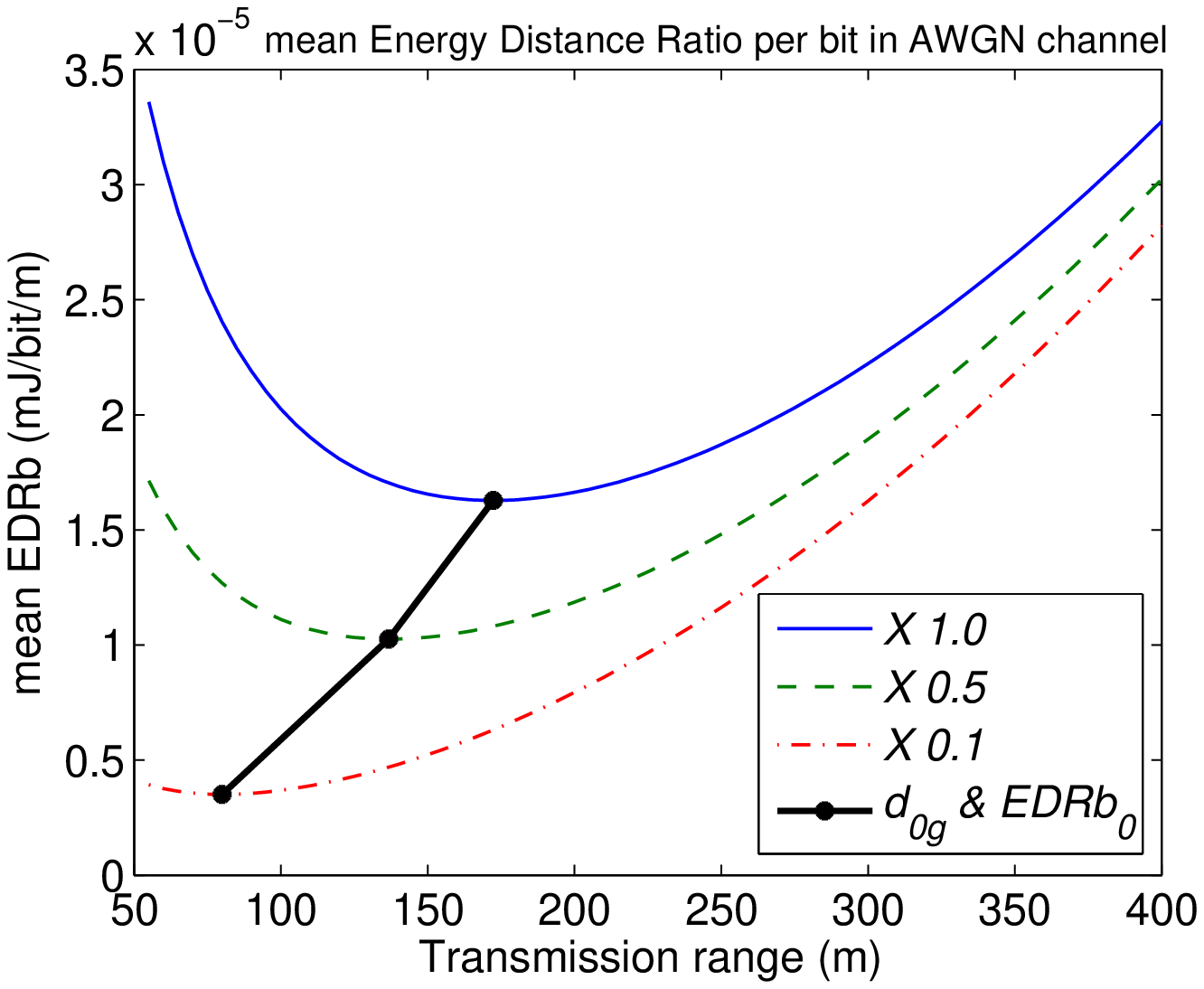}}
    \hfil
    \subfigure[Rayleigh flat fading channel]{
    \includegraphics[width=0.6\textwidth]{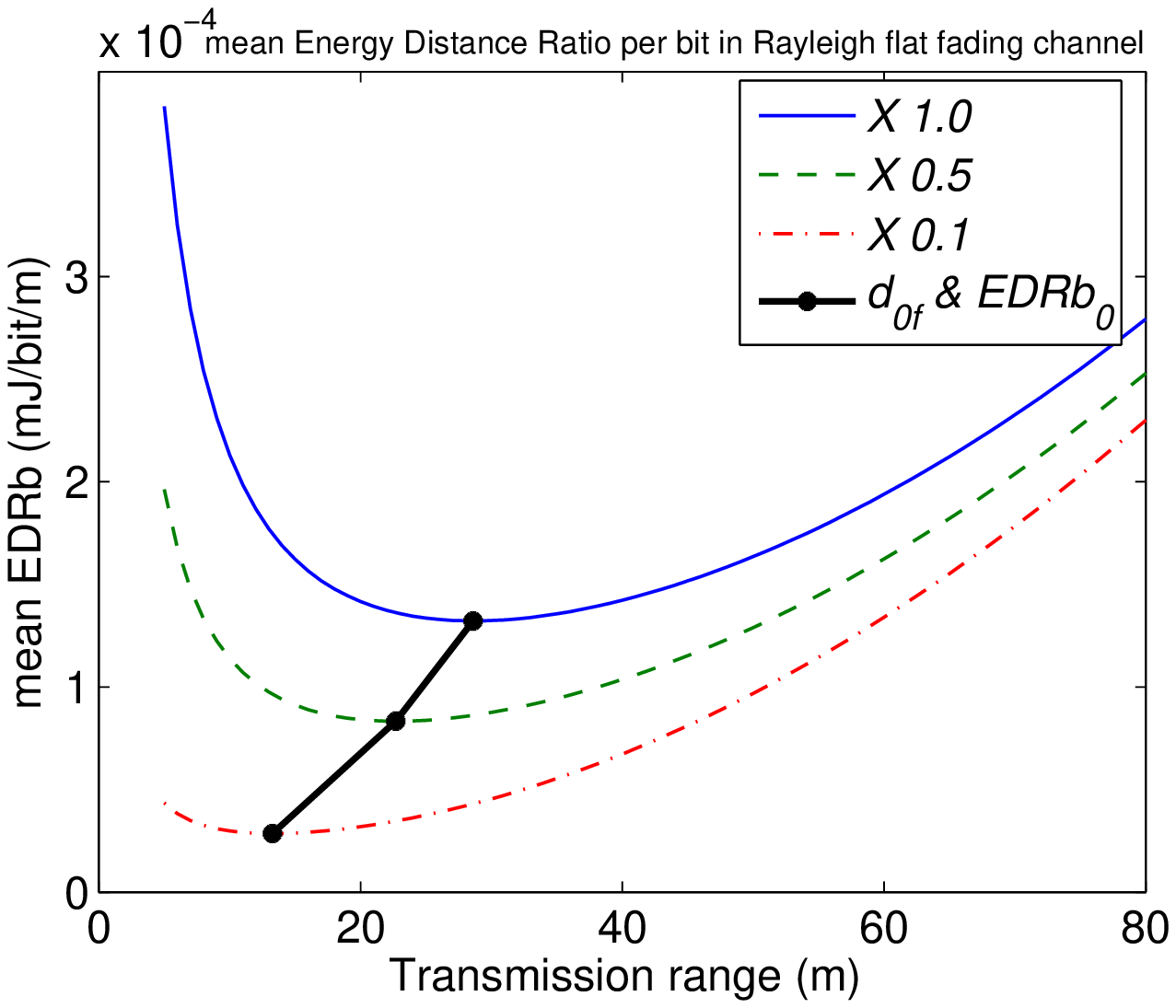}}
    \hfil
    \subfigure[Nakagami block fading channel $m=1$]{
    \includegraphics[width=0.6 \textwidth]{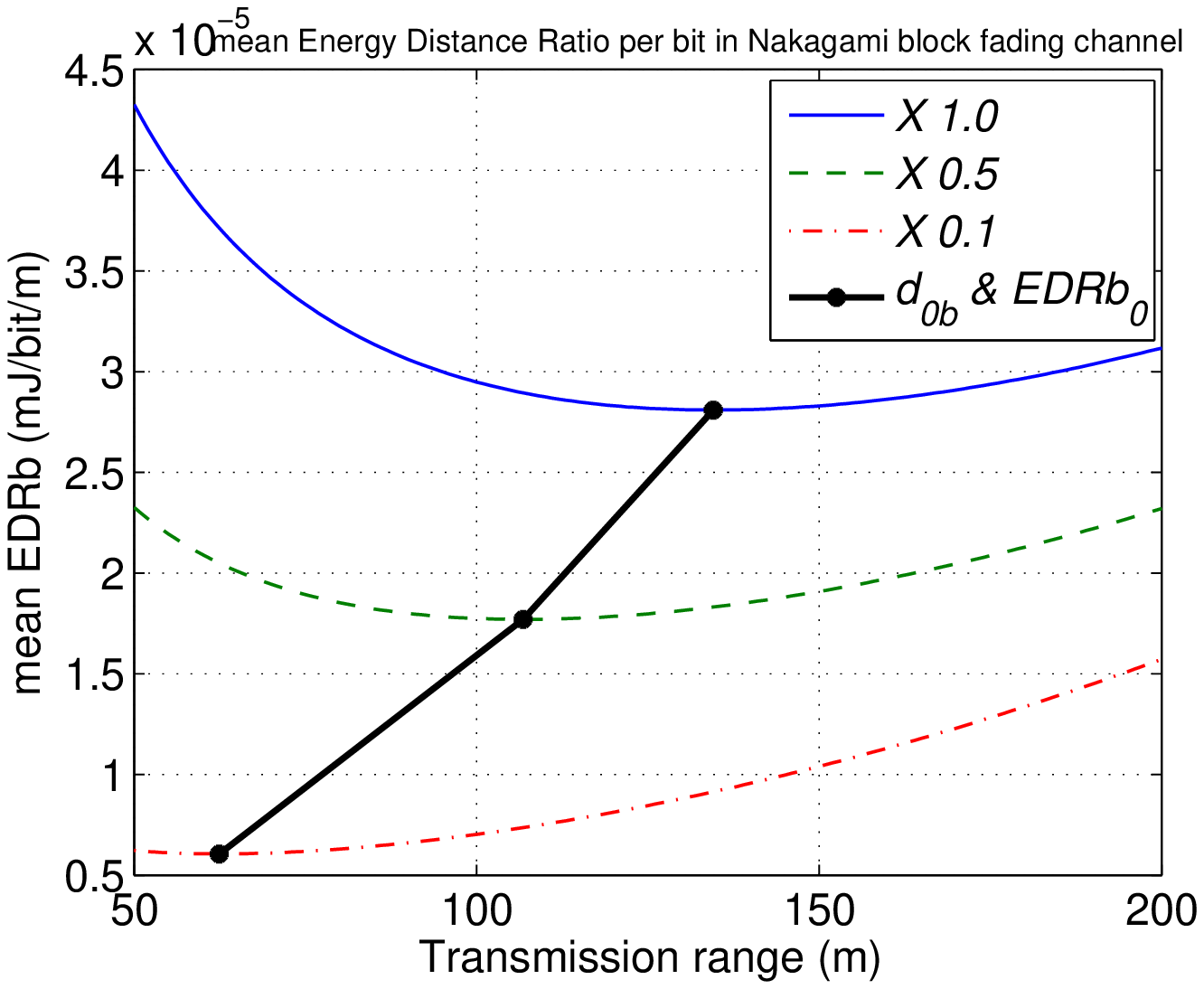}}
    \caption{Impact of the circuit power on the one-hop transmission energy
    performance given by $\overline{EDRb}$ as a function of $d$.}
     \label{fig:effect_circuit}
\end{figure*}
Fig.~\ref{fig:effect_circuit} shows the effect of circuity power on
$\overline{EDRb}$ and $d_0$, where the whole circuity powers
$P_{txElec}$, $P_{rxElec} $, $\alpha_{amp}$ and $P_{start}$ decrease
by the coefficients $0.5$ and $0.1$.  Since the reduction of
circuity powers results in the decrease of $P_0$ which leads to
shorten $d_0$. When the circuity powers are set to $0$, the shortest
hop distance has the high energy
efficiency~\cite{energy:Banerjee04}. Meanwhile, the energy
efficiency is improved with the reduction of circuity power. Hence,
the effect of circuity energy consumption should be considered in
the design of WSNs.
\paragraph*{Impact of the modulation}
\begin{figure*}[!t]
    \centering
    \subfigure[AWGN channel]{
    \includegraphics[width=0.6\textwidth]{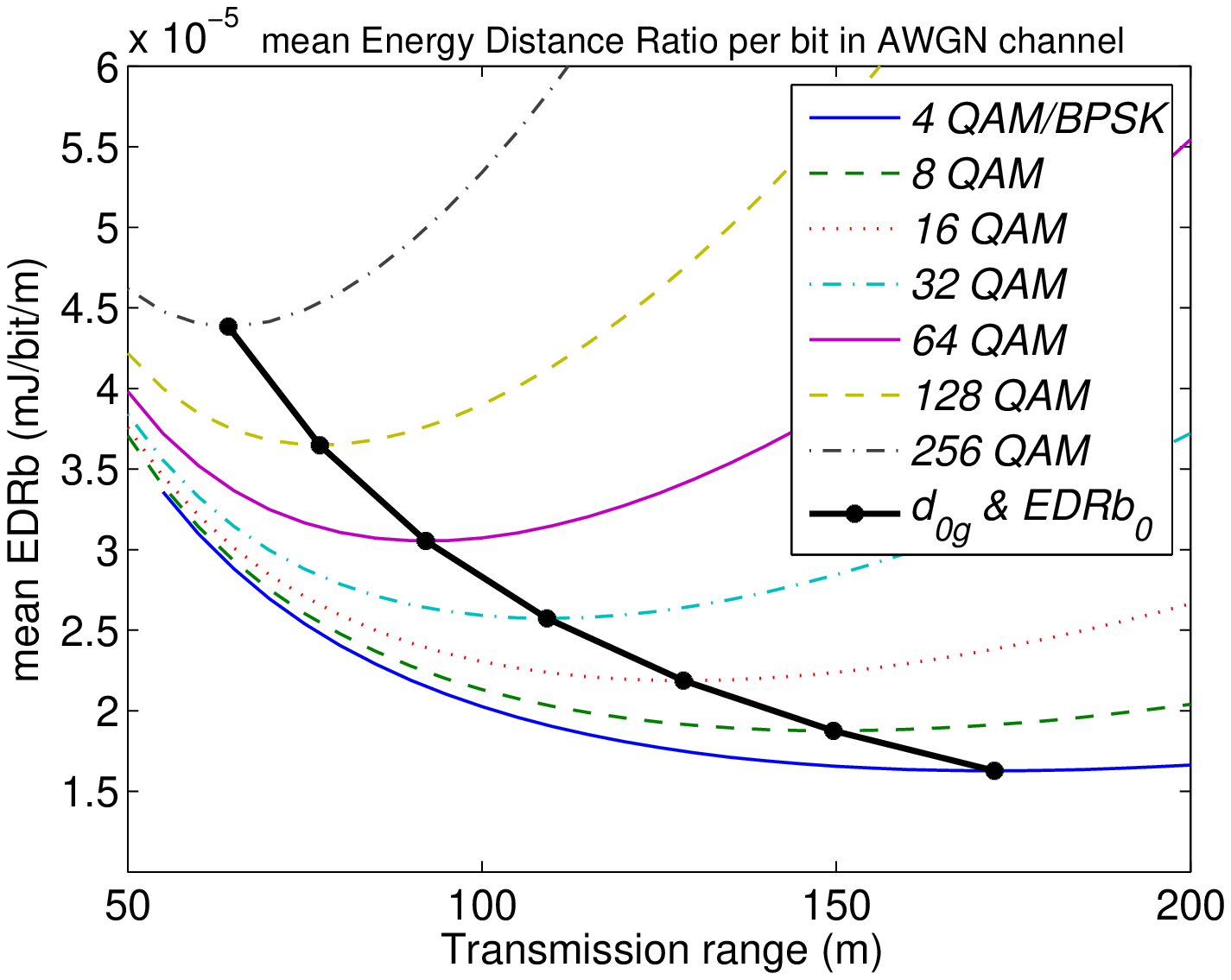}}
    \hfil
    \subfigure[Rayleigh flat fading channel]{
    \includegraphics[width=0.6\textwidth]{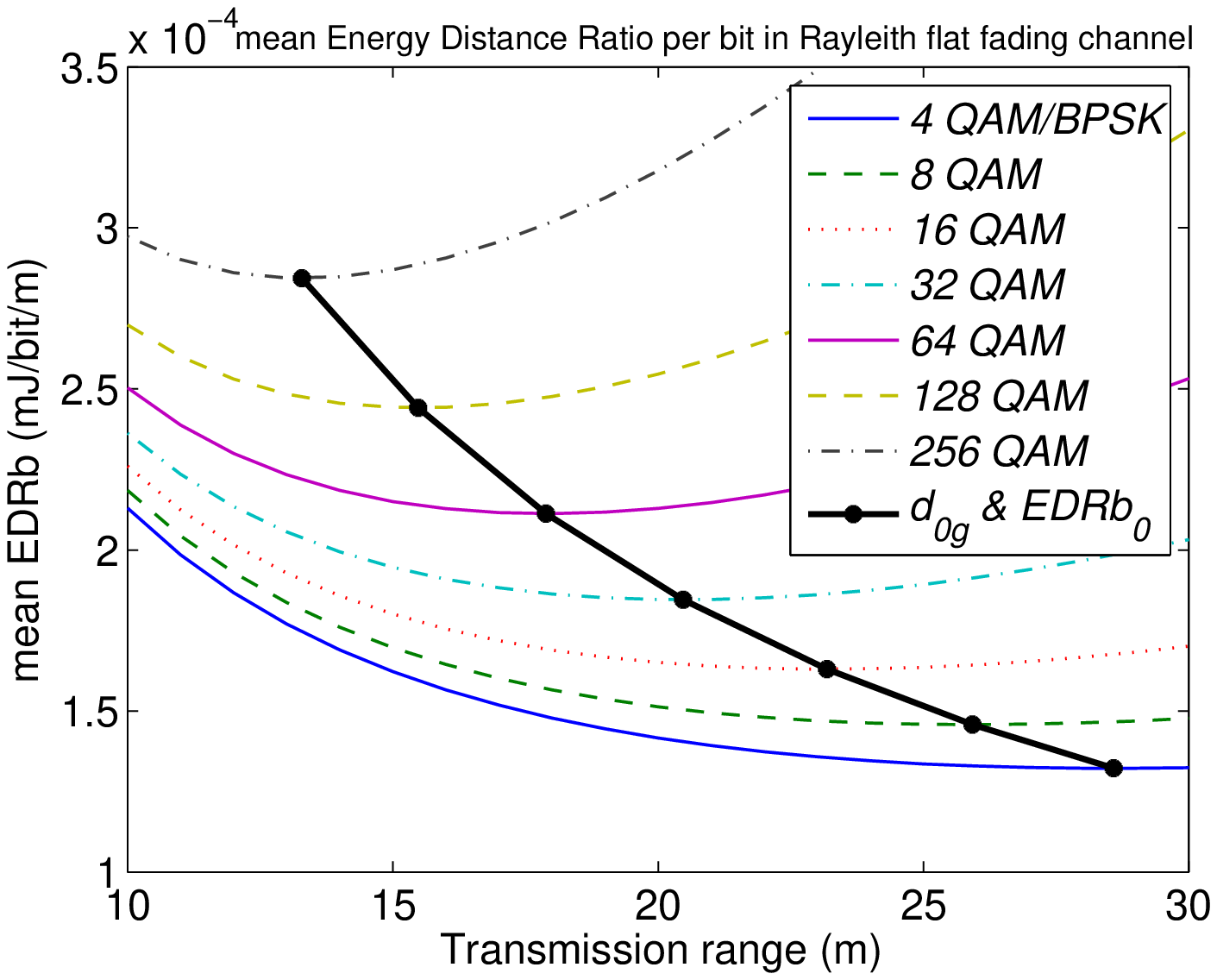}}
    \hfil
    \subfigure[Nakagami block fading channel $m=1$]{
    \includegraphics[width=0.6 \textwidth]{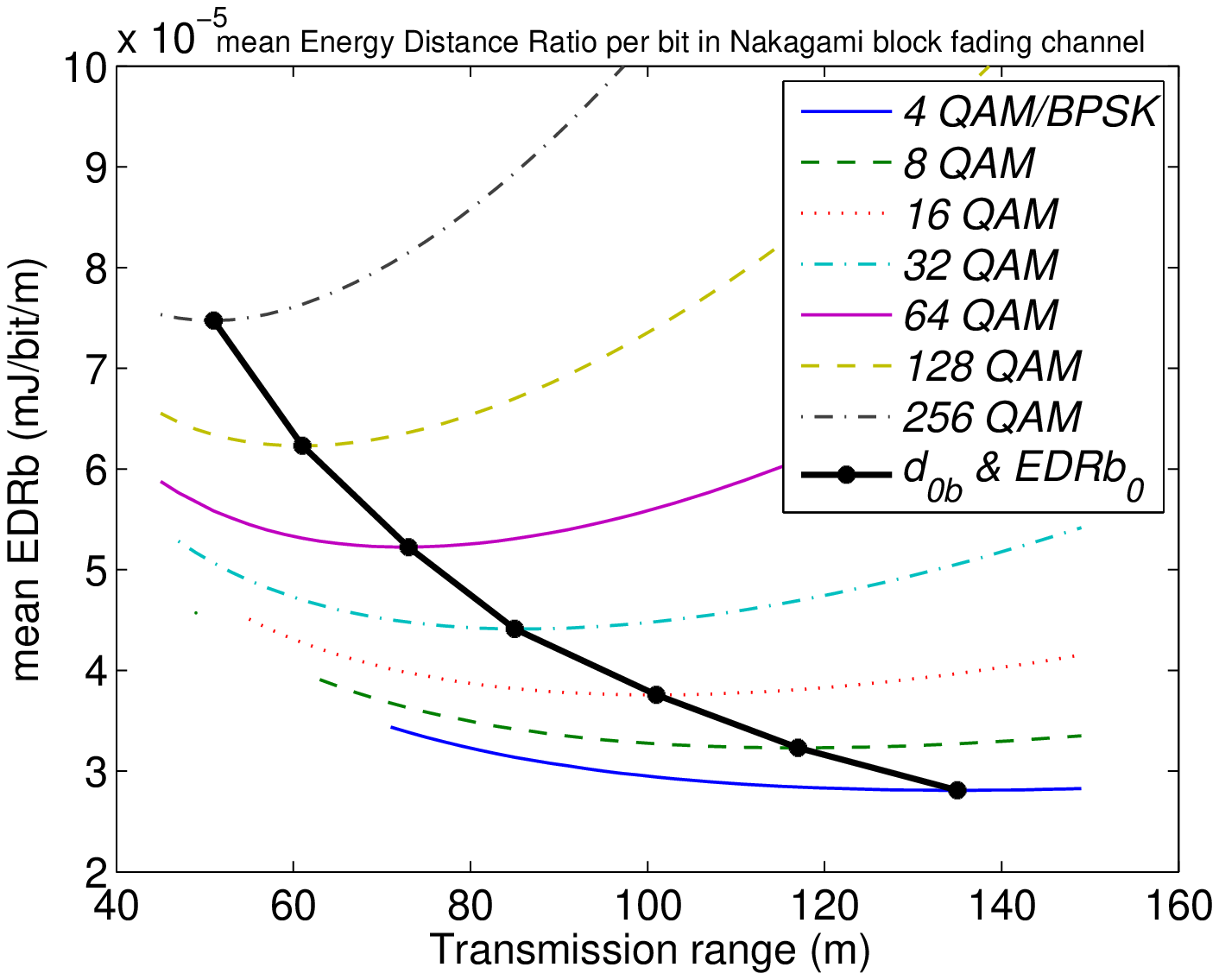}}
    \caption{Impact of the modulation on the one-hop transmission energy
    performance given by $\overline{EDRb}$ as a function of $d$.}
     \label{fig:effect_modulation}
\end{figure*}
The effect of modulation on the optimal $\overline{EDRb}$ for three
kinds of channel is shown in Fig.~\ref{fig:effect_modulation}. It
should be noted that the optimal $\overline{EDRb}$ monotonously
decreases while the optimal transmission range monotonously
increases with the decrease of the order of the modulation for the
three different channel types. 4QAM or BPSK are the most energy
efficient among the MQAM modulations which can be explained by BER.
BER increases with the order of the modulation for an identical SNR,
which leads to a reduced optimal transmission range. Due to the
reduction of the transmission range and duration, $E_c$ has a bigger
proportion in the total energy consumption, which results in the
increase of $\overline{EDRb}$.

\paragraph*{Impact of the packet size}
\begin{figure*}[!t]
\centering \subfigure[AWGN channel]{
\includegraphics[width=0.6\textwidth]{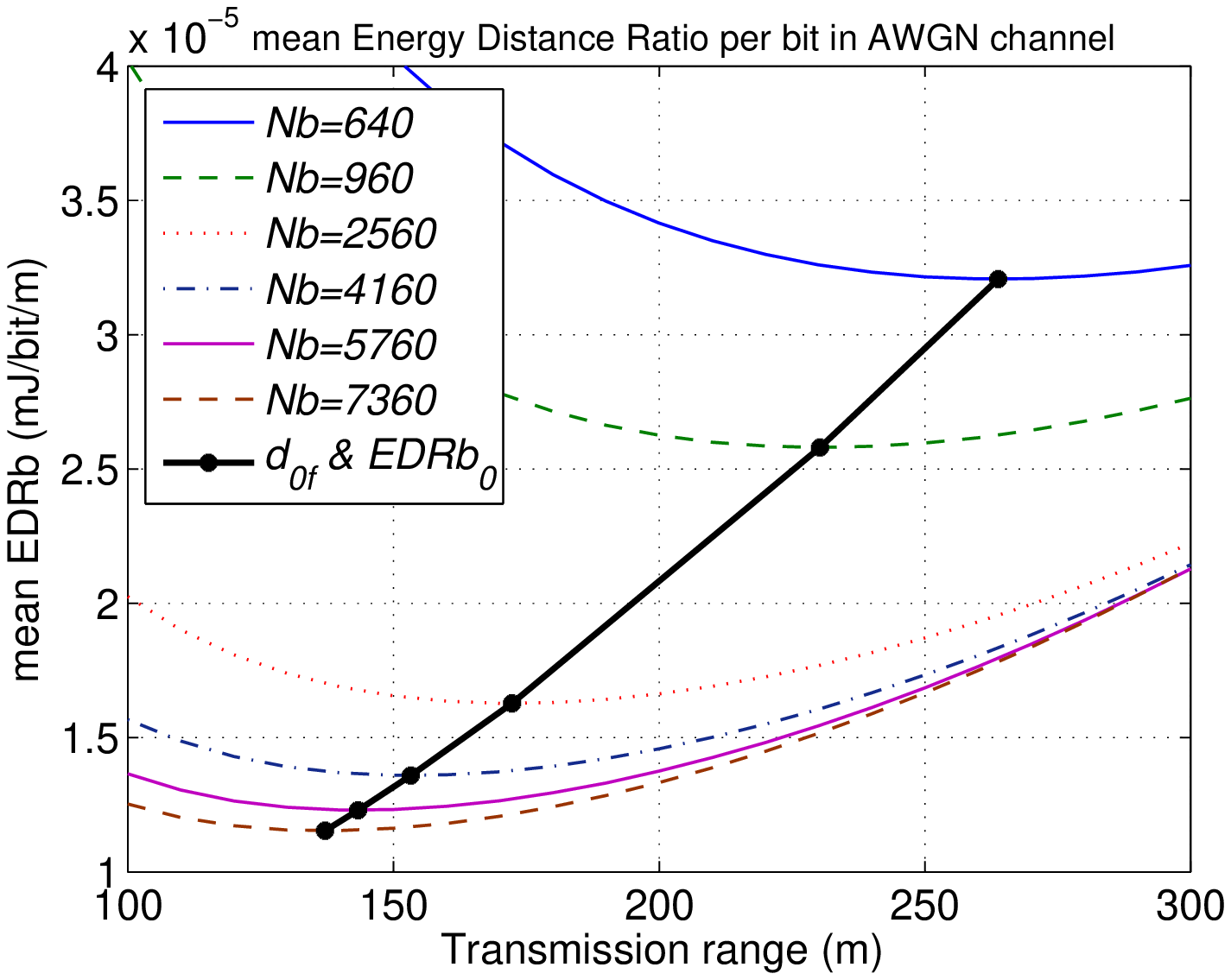}}
\hfil \subfigure[Rayleigh flat fading channel]{
\includegraphics[width=0.6\textwidth]{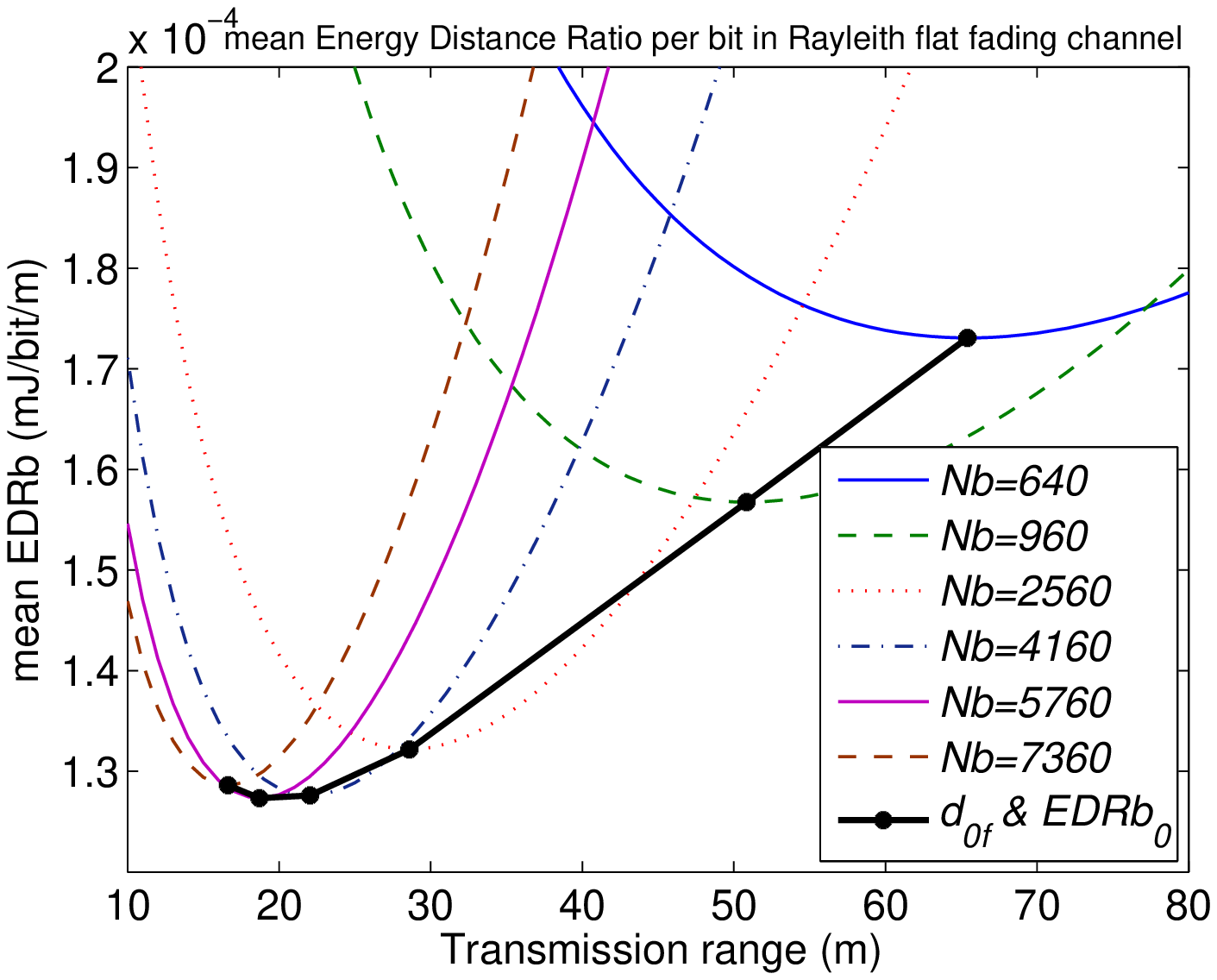}}
\hfil \subfigure[Nakagami block fading channel $m=1$]{
\includegraphics[width=0.6\textwidth]{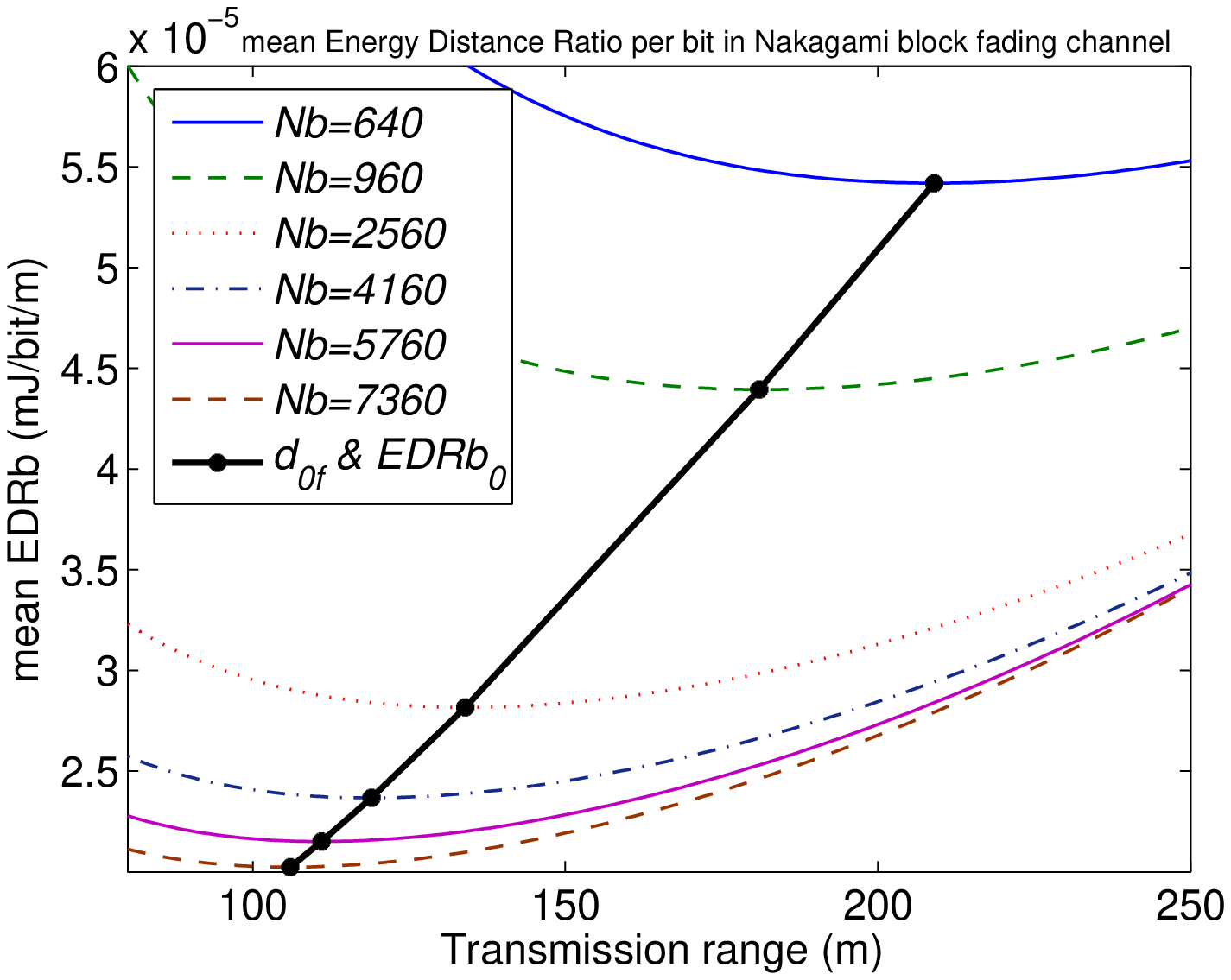}}
\caption{Impact of the packet size $N_b$ on the one-hop transmission
energy performance given by $\overline{EDRb}$ as a function of $d$.}
\label{fig:effect_Nb}
\end{figure*}
Fig.~\ref{fig:effect_Nb} shows how the optimal $\overline{EDRb}$
varies with $N_b$ and the corresponding optimal transmission range
for the three kinds of channel. In AWGN channel and Nakagami block
fading channel, the optimal $\overline{EDRb}$ and the optimal
transmission range decrease with the increase of $N_b$. In contrast,
for Rayleigh flat fading channel, there is an optimal $N_b$ that
originates from the trade-off between the variable transmission
energy ($K_1\cdot P_t$) and $E_c$. The proportion of $K_1\cdot P_t$
rises in the total energy consumption with the increase of $N_b$,
which trades off $E_c$. The increase of $N_b$ results in the
decrease of the link probability, which leads to the decrease of the
optimal transmission range. It can be deduced from
Fig.~\ref{fig:effect_Nb} that larger packets need less energy but
more hops and higher delays.


\paragraph*{Impact of the rate}
\begin{figure*}[!t]
    \centering
    \subfigure[AWGN channel]{
    \includegraphics[width=0.6\textwidth]{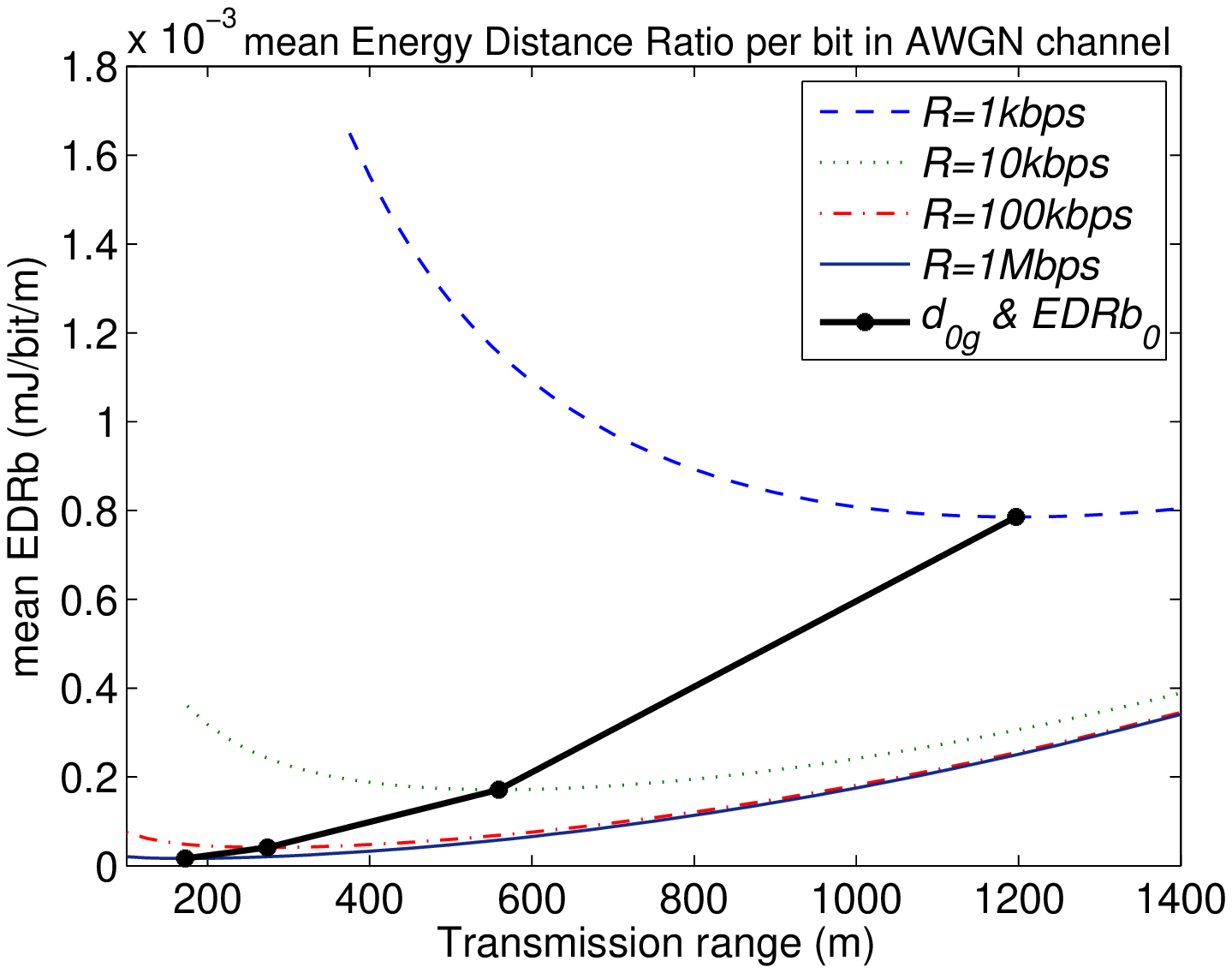}}
    \hfil
    \subfigure[Rayleigh flat fading channel]{
    \includegraphics[width=0.6\textwidth]{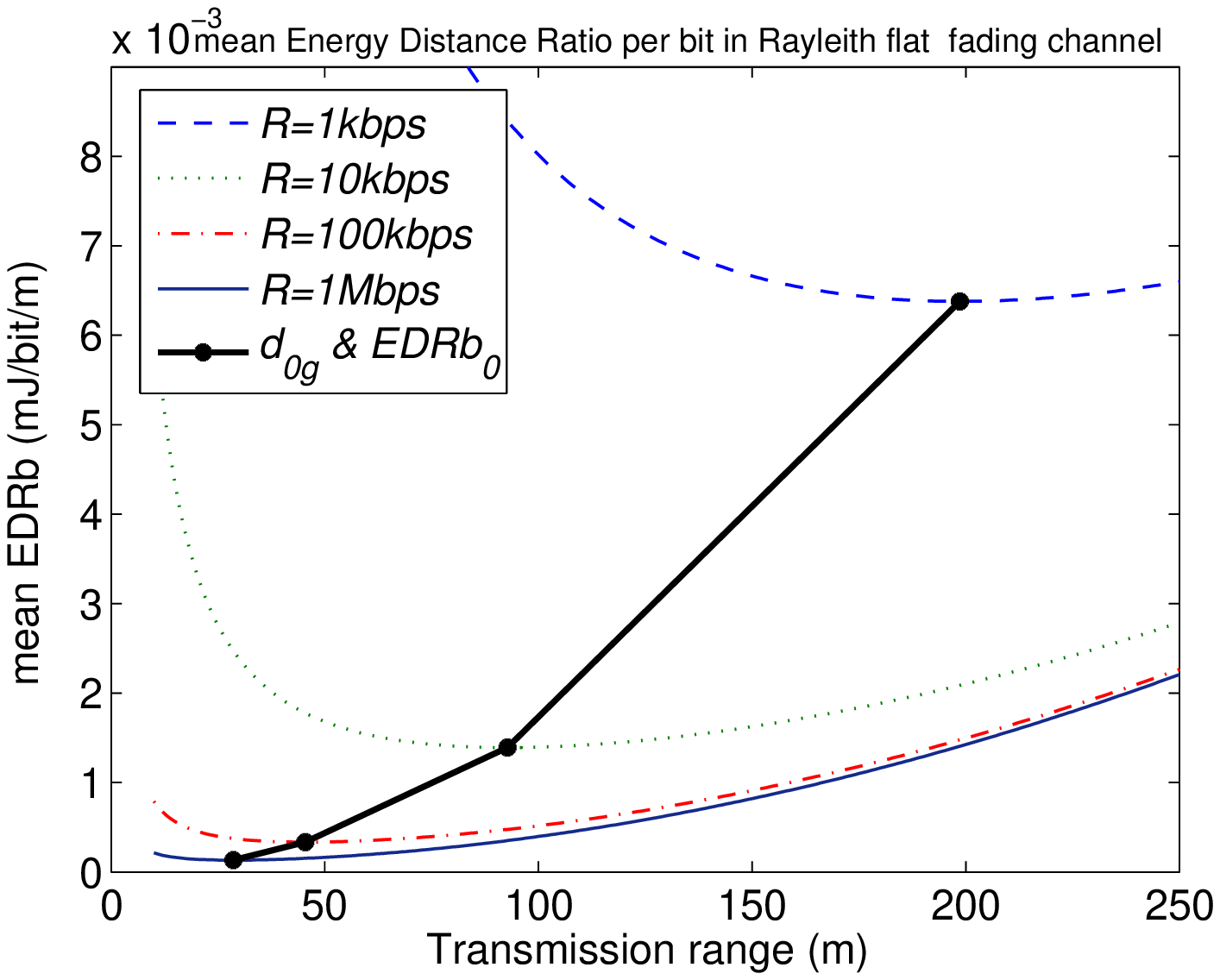}}
    \hfil
    \subfigure[Nakagami block fading channel $m=1$]{
    \includegraphics[width=0.6 \textwidth]{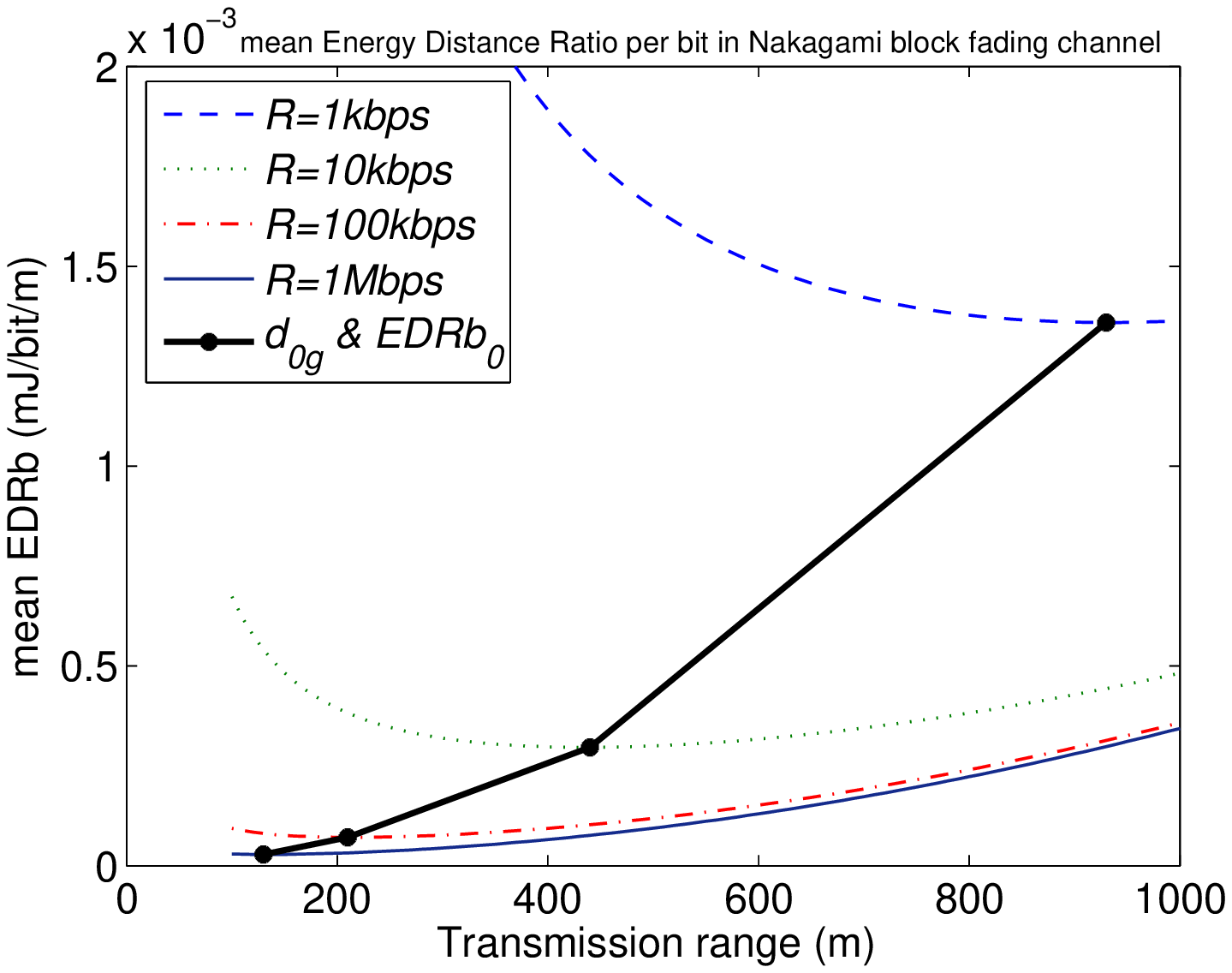}}
    \caption{Impact of the rate $R$ on the one-hop transmission energy
    performance given by $\overline{EDRb}$ as a function of $d$.}
    \label{fig:effect_Rtx}
 \end{figure*}
In Fig.~\ref{fig:effect_Rtx}, the increase of transmission rate
leads to the decrease of the link probability according to
Eq.~\eqref{eq:gamma} which brings forth the reduction of the optimal
transmission range. Meanwhile, the reduction of the total energy
consumption results in the decrease of the optimal
$\overline{EDRb}$.

\section{Multi-hop Transmission: Energy Efficiency and Delay}
\label{sec:low_bound} In this section, a multi-hop transmission
along a homogeneous linear network is considered. Nodes are aligned
because a transmission using properly aligned relays is more energy
efficient than a transmission where the same relays do not belong to
the straight line defined by the source and the destination. In this
section, we first prove that the transmission along equidistant hops
is the best way for saving energy in a homogeneous linear network.
Next, the optimal number of hops over a homogeneous linear network
is derived for a given transmission distance according to the
optimal one-hop transmission distance. Finally, a lower bound on the
energy efficiency and its delay is obtained for the considered
multi-hop transmission.

\subsection{Minimum mean total energy consumption}
\begin{theorem}
\label{the:minenergy} In a homogeneous linear network, a source node
$x$ sends a  packet of $N_b$ bits to a destination node $x'$ using
$n$ hops. The distance between $x$ and $x'$ is $d$. The length of
each hop is $d_1$, $d_2$, \ldots, $d_n$ respectively and the average
EDRb is denoted $\overline{EDRb}(d)$. The minimum mean total energy
consumption $\overline{Etot}_{min}$ is obtained for if and only if
$d_1=d_2=\ldots=d_n$:
\begin{equation}
\overline{Etot}_{min} = N_b\cdot \overline{EDRb}(d/n)\cdot d.
\end{equation}
\end{theorem}

\begin{proof} The mean energy consumption for each hop of index $m$ is set to
$\overline{E}_m = N_b\cdot\overline{EDRb}(d_m)\cdot d_m,\ m = 1, 2,
\ldots, n$. Since each hop is independent from the other hops, the
mean total energy consumption is \[\overline{Etot} =
\overline{E}_1+\overline{E}_2+\ldots+\overline{E}_n.\]  Hence, the
problem of finding the minimum mean total energy consumption can be
rewritten as:
\begin{align} &\text{minimize} &&\overline{Etot}
\notag \\
&\text{subject to}   && d_1+d_2+\ldots +d_n = d.\notag
\end{align}
Set
\[F=\overline{E}_1+\overline{E}_2+\ldots+\overline{E}_n+\lambda(d_1+d_2+\ldots
+d_n-d), \] where $\lambda\neq0$ is the Lagrange multiplier.
According to the method of the Lagrange multipliers, we obtain
\begin{equation}
\label{eq:lagrange_energy} \left\{
\begin{split}
&\frac{\partial{\overline{E_1}}}{\partial{d_1}}+ \lambda =0\\
&\frac{\partial{\overline{E_2}}}{\partial{d_2}}+ \lambda =0\\
&\ldots\\
&\frac{\partial{\overline{E_n}}}{\partial{d_n}}+ \lambda =0\\
&d_1+d_2+\ldots +d_n = d
\end{split}
\right.
\end{equation}
Eq.~\eqref{eq:lagrange_energy} shows that the minimum value of $F$
is obtained in the case
$\frac{\partial{\overline{E_1}}}{\partial{d_1}}=\frac{\partial{\overline{E_2}}}{\partial{d_2}}
=\ldots =\frac{\partial{\overline{E_n}}}{\partial{d_n}} = -\lambda$.
Moreover, in a homogeneous linear network, the properties of each
node are identical. Therefore,
\[\frac{\partial{\overline{E_m}}}{\partial{d_m}}=\frac{\partial{\overline{E}}}{\partial{d}}\bigg|_{d=d_m}\]
where $m=1,2,\ldots,n $. Because
$\frac{\partial{\overline{E}}}{\partial{d}}$ is a monotonic
increasing function of $d$ when the path-loss exponent follows
$\alpha\geq2$, the unique solution of Eq.~\eqref{eq:lagrange_energy}
is $d_1=d_2=\ldots=d_n=\frac{d}{n}$. Finally, we obtain:\[
\overline{Etot}_{min} = N_b\cdot \overline{EDRb}(d/n)\cdot d.
\] \end{proof}

\subsection{Optimal number of hops}
\label{subsec:opti-hop} Based on Theorem~\ref{the:minenergy} and the
analysis in Section~\ref{sec:one hop}, the optimal hop number can be
calculated from the transmission distance $d$ and the optimal
one-hop transmission distance $d_0$. When $d/d_0$ is an integer,
$[d/d_0]$ is the optimal hop number $N_{hop0}$ as each hop has the
minimum $\overline {EDRb}$ according to Theorem~\ref{the:minenergy}.
When $d/d_0$ is not an integer, setting $\lfloor d/d_0 \rfloor = n$,
the optimal hop number is $N_{hop0}= n\ or\ n+1$, which can be
decided by:
\begin{equation}
\label{eq:hops_opti}
Min\left\{ \overline{EDRb}(d/n), \overline{EDRb}(d/(n+1))\right\}
\end{equation}
where $\lfloor x\rfloor$ provides the largest integer value smaller
or equal to $x$. The transmission range of each hop is now
$d/N_{hop0}$.

\subsection{Lower bound on $\overline{EDRb}$ and its delay}
Substituting the formula $P_0$ and $d_0$ in three kinks of channel
into~\eqref{eq:EDRb} yields:
\begin{equation}
\overline{EDRb} = \frac{E_c+K_1 P_0}{d_0 p_l(P_0,d_0)}
\label{eq:EDRB_lowbound}
\end{equation}
Equation \eqref{eq:EDRB_lowbound} provides the exact lower bound of
$\overline{EDRb}$ on the basis of Theorem~\ref{the:minenergy} and
the analyzes of section~\ref{sec:one hop} for a multi-hop
transmission using $n$ hops. Its corresponding end-to-end delay is
computed as:
\begin{equation}
\label{eq:delay_mh} \overline D_{opt} = N_{hop0}\cdot \overline
D_{ch}
\end{equation}
where  $\overline D_{ch}$ is the one-hop transmission delays and
$ch$ respectively stands for Eq.~\eqref{eq:delay_oh_gaussian} with
respect to AWGN channel, Eq.~\eqref{eq:delay_oh_fading} for Rayleigh
flat fading channel and Eq.~\eqref{eq:delay_oh_block} for Nakagami
block fading channel.

\begin{figure*}[!t]
    \centering
     \subfigure[AWGN channel]{
    \includegraphics[width=0.6\textwidth]{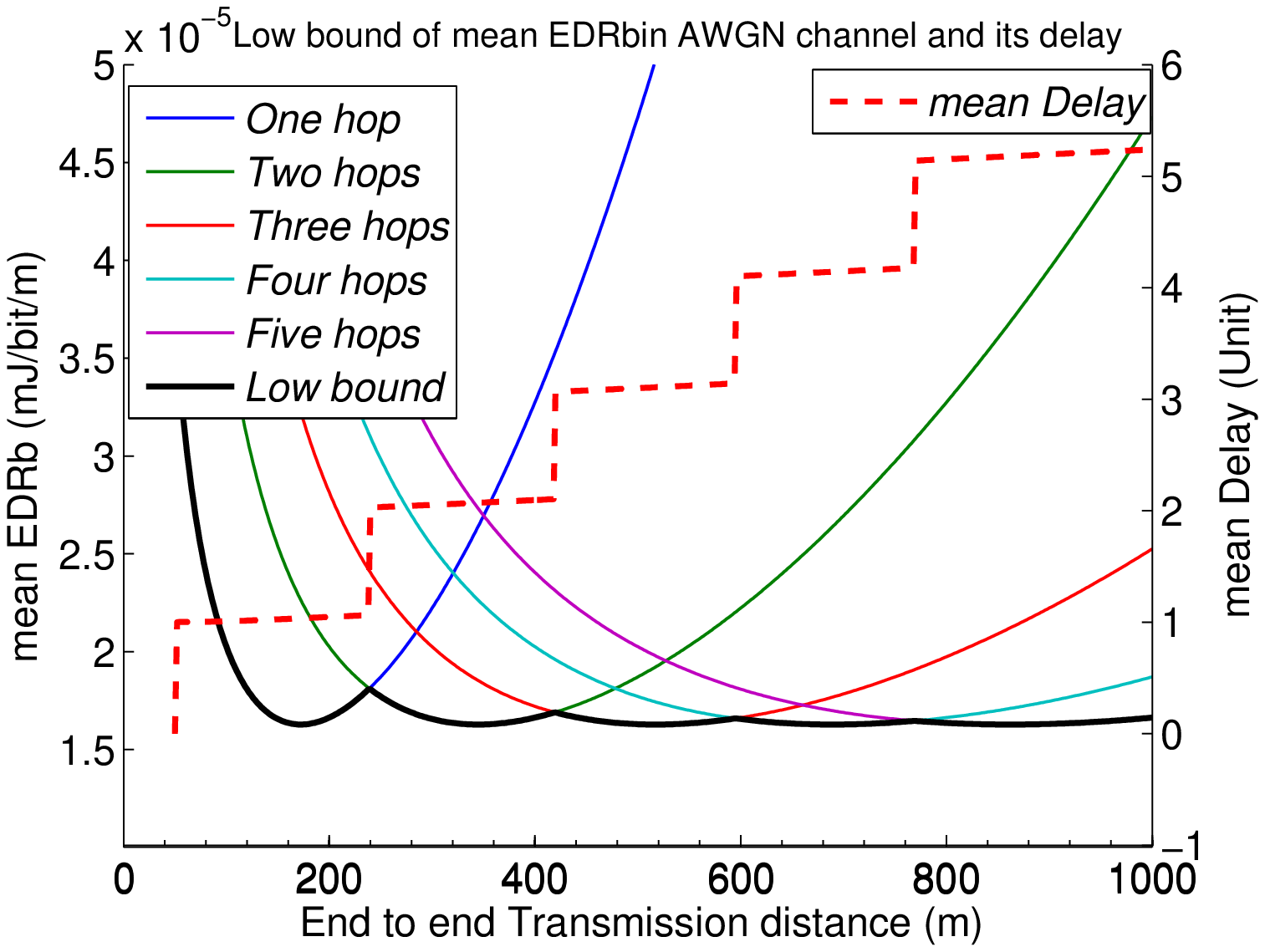}}
    \hfil
    \subfigure[Rayleigh flat fading channel]{
    \includegraphics[width=0.6\textwidth]{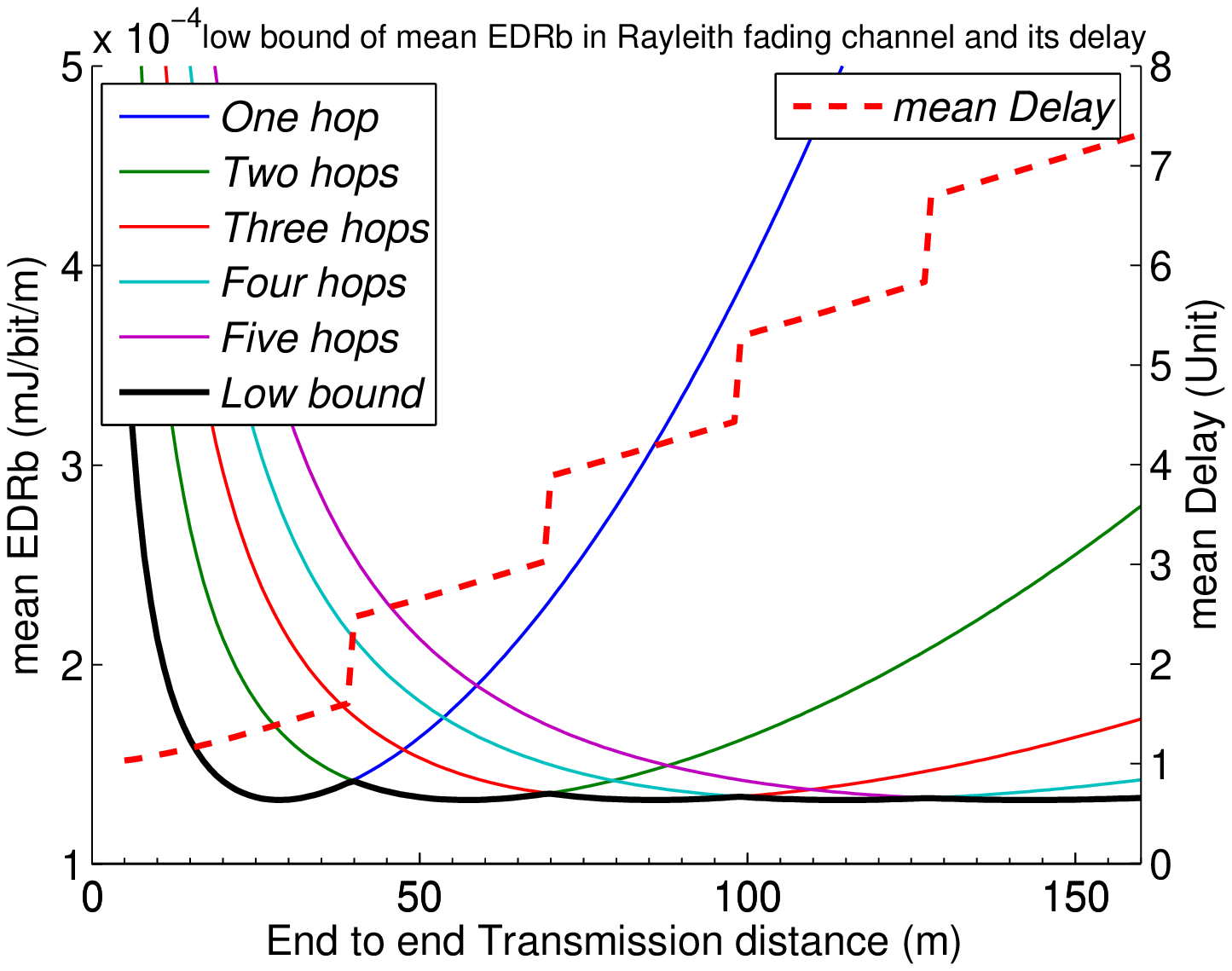}}
    \hfil
    \subfigure[Nakagami block fading channel]{
    \includegraphics[width=0.6\textwidth]{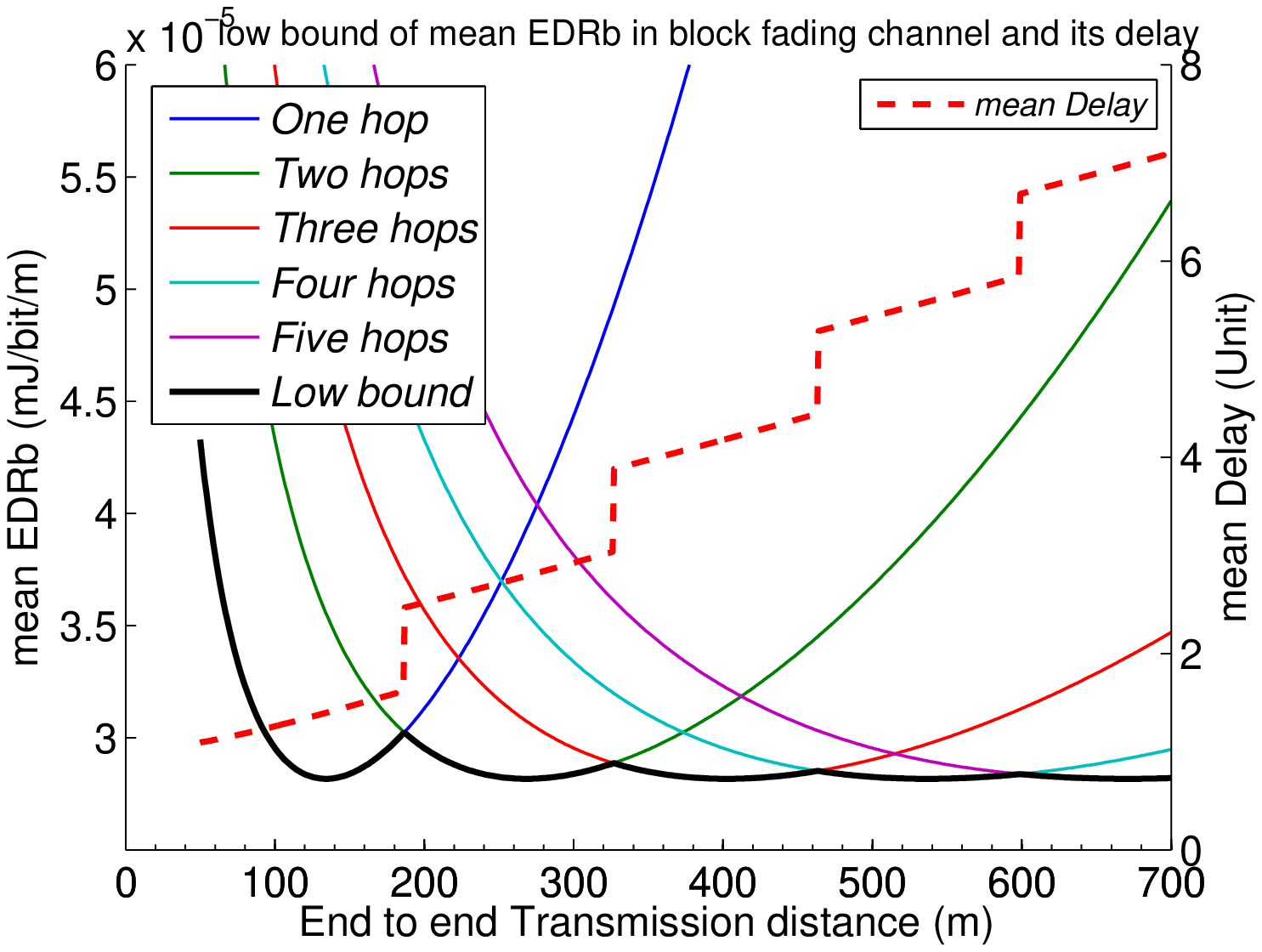}}
    \caption{Theoretical lower bound on $\overline{EDRb}$ and the corresponding mean delay for a homogeneous linear network}
\label{fig:EDRb_low_bound}
\end{figure*}
Fig.~\ref{fig:EDRb_low_bound} represents the theoretical lower bound
on $\overline{EDRb}$ and its corresponding mean delay over AWGN,
Rayleigh flat fading and Nakagami block fading channel. The
corresponding mean delay is obtained by Eq.~\eqref{eq:delay_mh}. It
can be noticed that the minimum value of $\overline{EDRb}$ can be
reached by following for each hop the optimum one-hop distance. It
is shown in section~\ref{sec:sim} that this lower bound is also
valid for 2-dimensional Poisson distributed networks using
simulations.

\section{Energy-Delay Trade-off}
\label{sec:en_delay}A trade-off between energy and delay exists. For
instance when considering long range transmissions, a direct
single-hop transmission needs a lot of energy but yields a shorter
delay while a multi-hop transmission uses less energy but suffers
from an extended delay as shown in Fig.~\ref{fig:EDRb_low_bound}.
This section concentrates on the analyses of the energy-delay
trade-off for both the one-hop and the multi-hop transmissions.

\subsection{Energy-delay trade-off for one-hop transmissions}
\label{subsec:tradeoff_onehop}
\begin{figure*}[!t]
    \centering
    \subfigure[AWGN channel $P_{0g}=535.87mW$]{
    \includegraphics[width=0.6\textwidth]{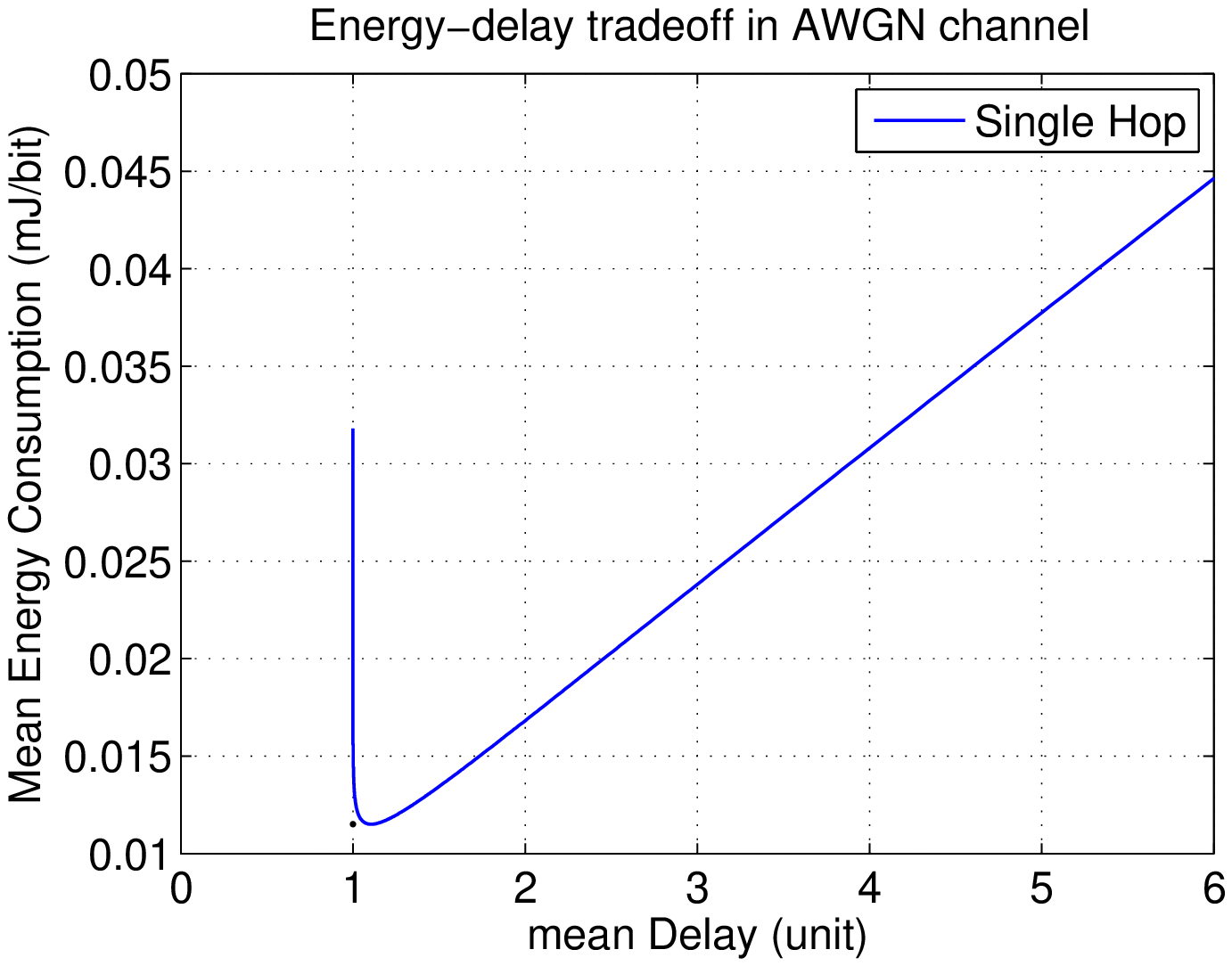}}
    \hfil
    \subfigure[Rayleigh flat fading channel $P_{0f} =734.12 mW$]{
    \includegraphics[width=0.6\textwidth]{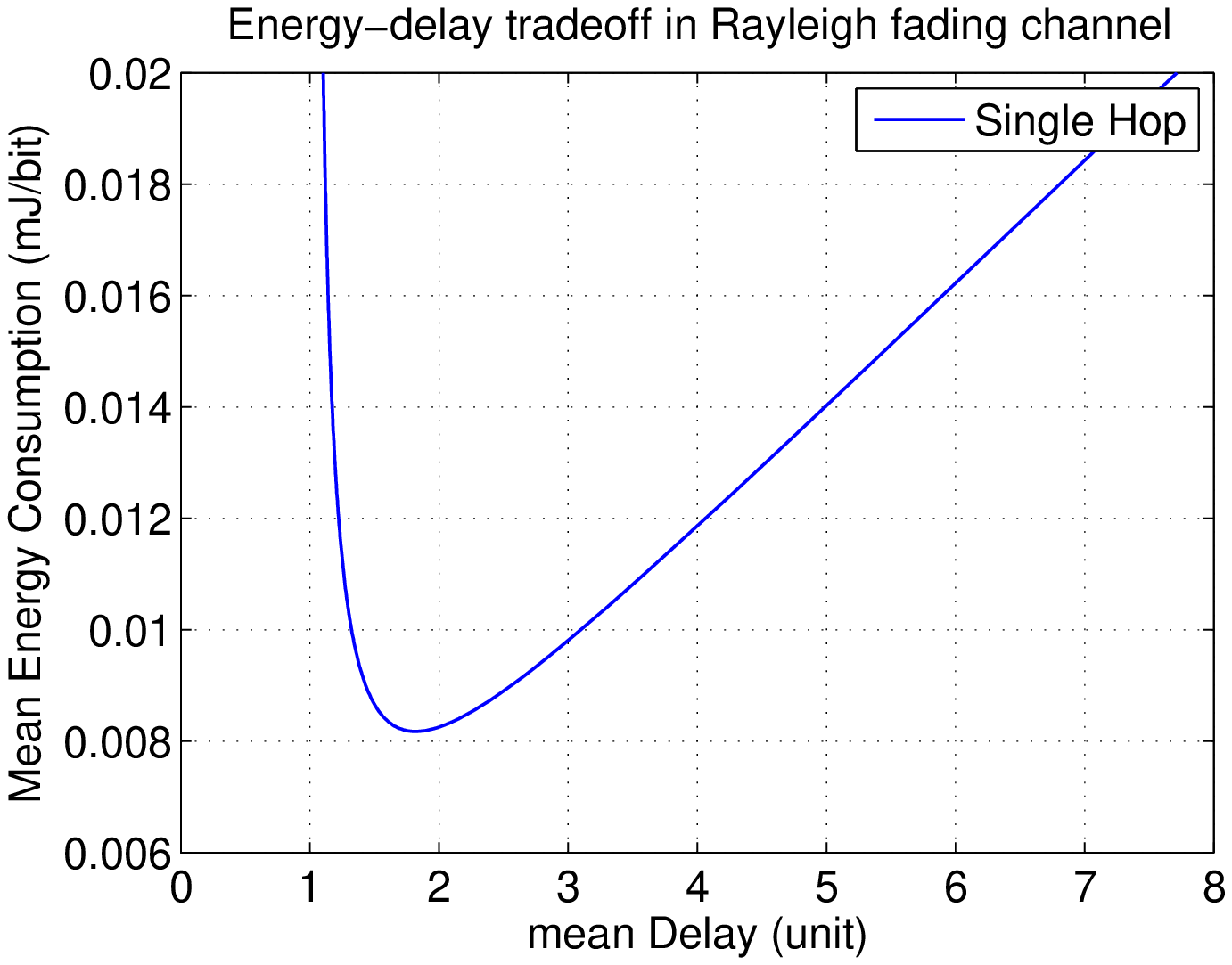}}
    \hfil
    \subfigure[Nakagami block fading channel $P_{0b}=535.87mW$]{
    \includegraphics[width=0.6\textwidth]{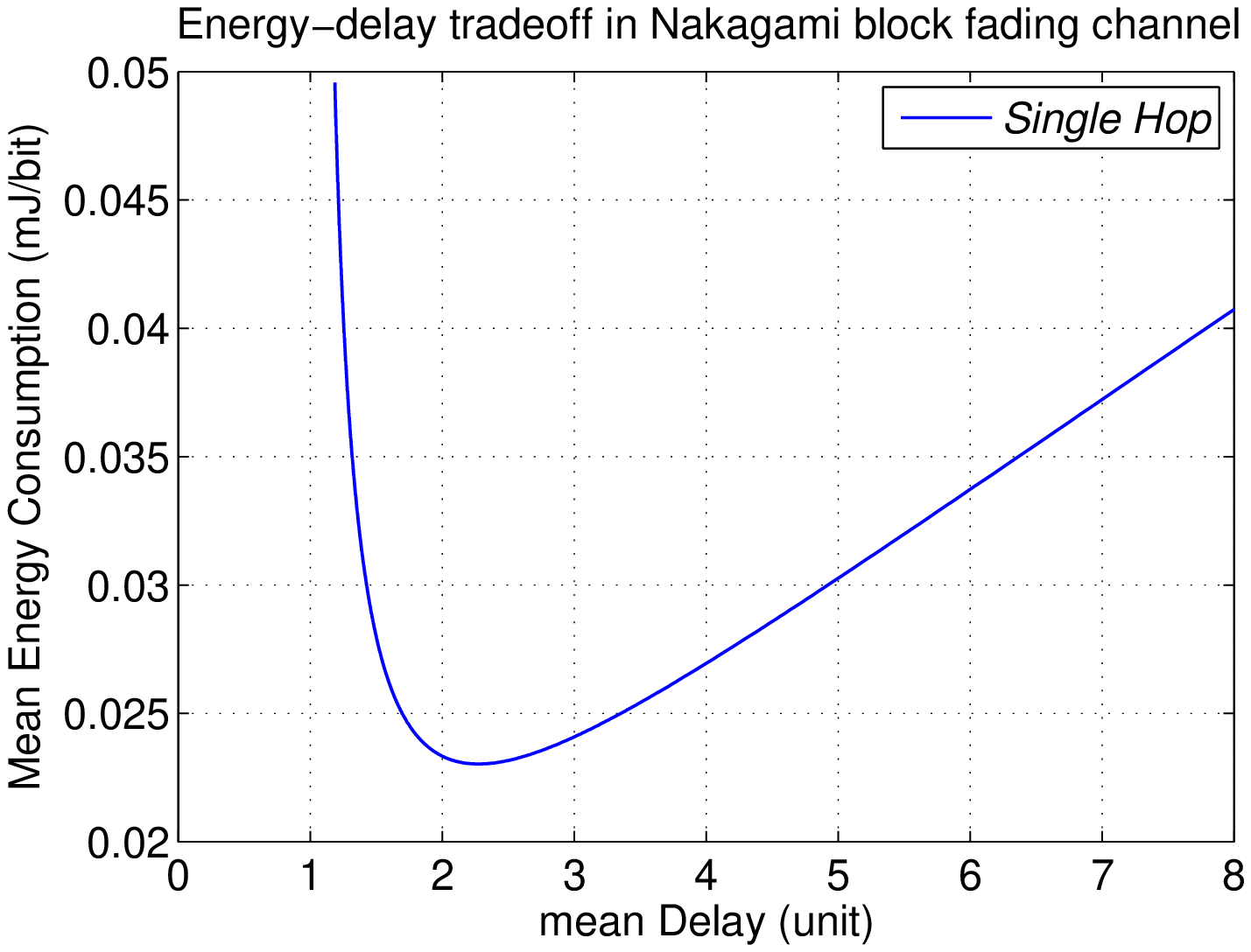}}
    \caption{Energy delay trade-off for the one-hop transmission}
    \label{fig:ed_to_oh}
\end{figure*}
Fig.~\ref{fig:ed_to_oh} shows the energy-delay trade-off of one-hop
transmission at a given distance $d=380m$ in AWGN and Nakagami block
fading channel and $d=50m$ in Rayleigh flat fading channel. These
three curves are obtained by varying the transmission power under
this fixed transmission range. The mean delay is computed
respectively using
Eq.~\eqref{eq:delay_oh_gaussian},~\eqref{eq:delay_oh_fading}
and~\eqref{eq:delay_oh_block} and the mean energy consumption is
calculated by Eq.~\eqref{eq:EDRb} over each kind of channel. The
lowest points on the three curves represent the minimum energy
consumptions possible for each type of channel. They correspond
respectively to the energy-optimum power values $P_{0g}=535.87mW$,
$P_{0f} =734.12 mW$, $P_{0b}=535.87mW$ which are the same than the
ones obtained with Eq.~\eqref{eq:P0_d_g},~\eqref{eq:P0_d}
and~\eqref{eq:P0_d_b} in section~\ref{sec:one hop}.

In Fig.~\ref{fig:ed_to_oh}, each curve can be analyzed according to
the transmission power used to obtain the energy-delay value. On
each curve, the points on the left of the minimum energy point are
obtained with transmission powers higher than the energy-optimum
power value $P_0$. The points on the right (i.e. experiencing higher
delays) are obtained for transmission powers smaller than the
energy-optimum power value $P_0$.

When $P_t$ is increasing and $P_t>P_0$, the energy consumption
increases drastically while the mean delay decreases as the link
gets more and more reliable. On the contrary, when $P_t$ is
decreasing and $P_t<P_0$, the energy consumption is increasing with
a slower pace while the mean delay increases as the link gets more
and more unreliable. More and more retransmissions are here
performed, using more energy and increasing the one-hop transmission
delay.

\subsection{Energy-delay trade-off for multi-hop transmissions}
\label{subsec:bound_tradeoff} In section~\ref{subsec:opti-hop}, the
lower bound on the energy efficiency for a given transmission
distance and its corresponding delay are analyzed determining the
point of minimum energy and largest delay for a multi-hop
transmission. However, in some applications subject to delay
constraints, the energy consumption can be raised to diminish the
transmission delay. Therefore, the energy-delay trade-off for
multi-hop transmissions is analyzed in the following. To determine
the energy-delay trade-off for multi-hop transmissions, we still
consider a linear homogeneous network and show in Theorem
\ref{the:mindelay} that the minimum mean delay is also obtained for
equidistant hops.

\begin{figure*}[!t]
    \centering
    \subfigure[AWGN channel $d = 380m$]{
    \includegraphics[width=0.6\textwidth]{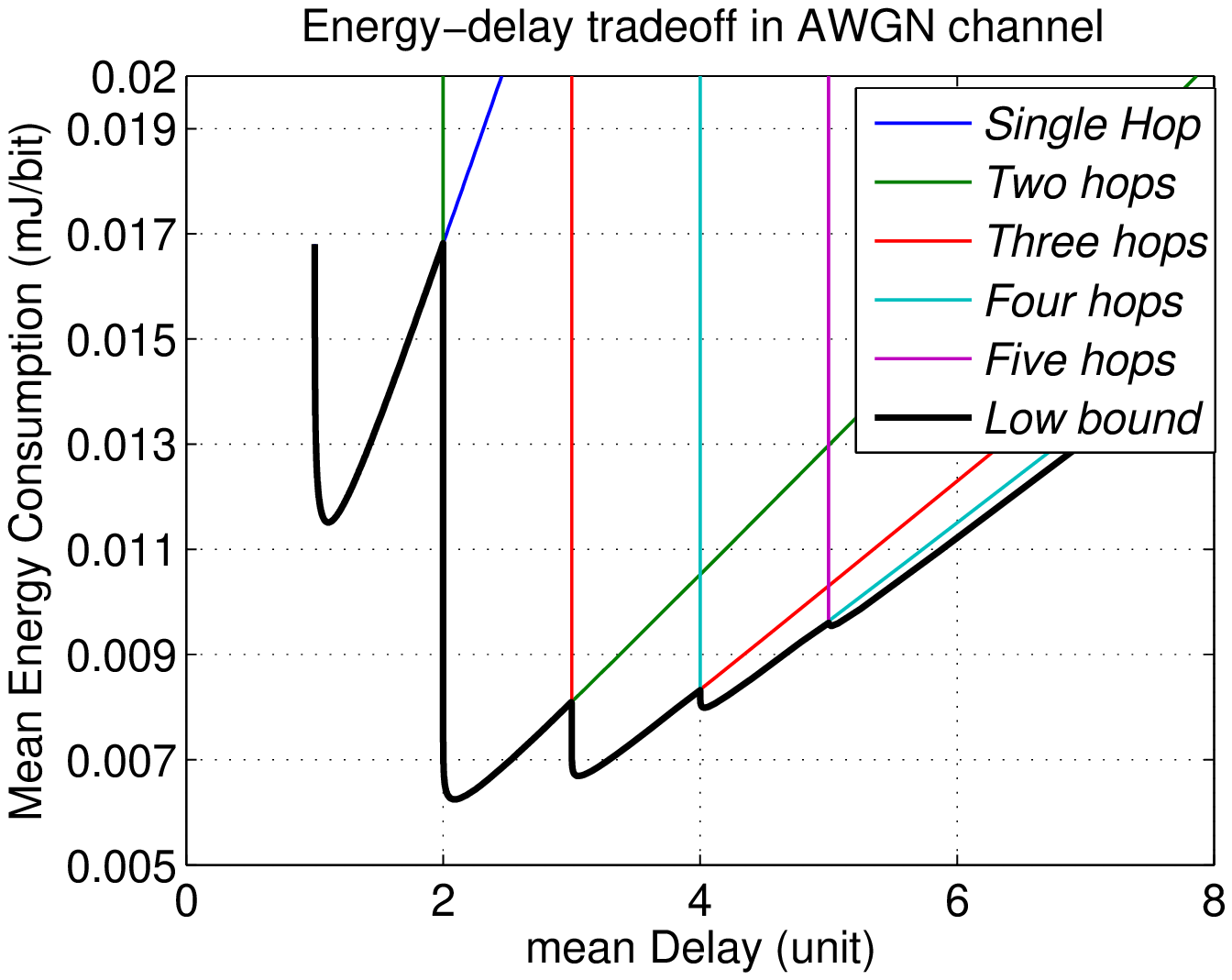}}
    \hfil
    \subfigure[Rayleigh flat fading channel $d = 50m$]{\includegraphics[width=0.6\textwidth]{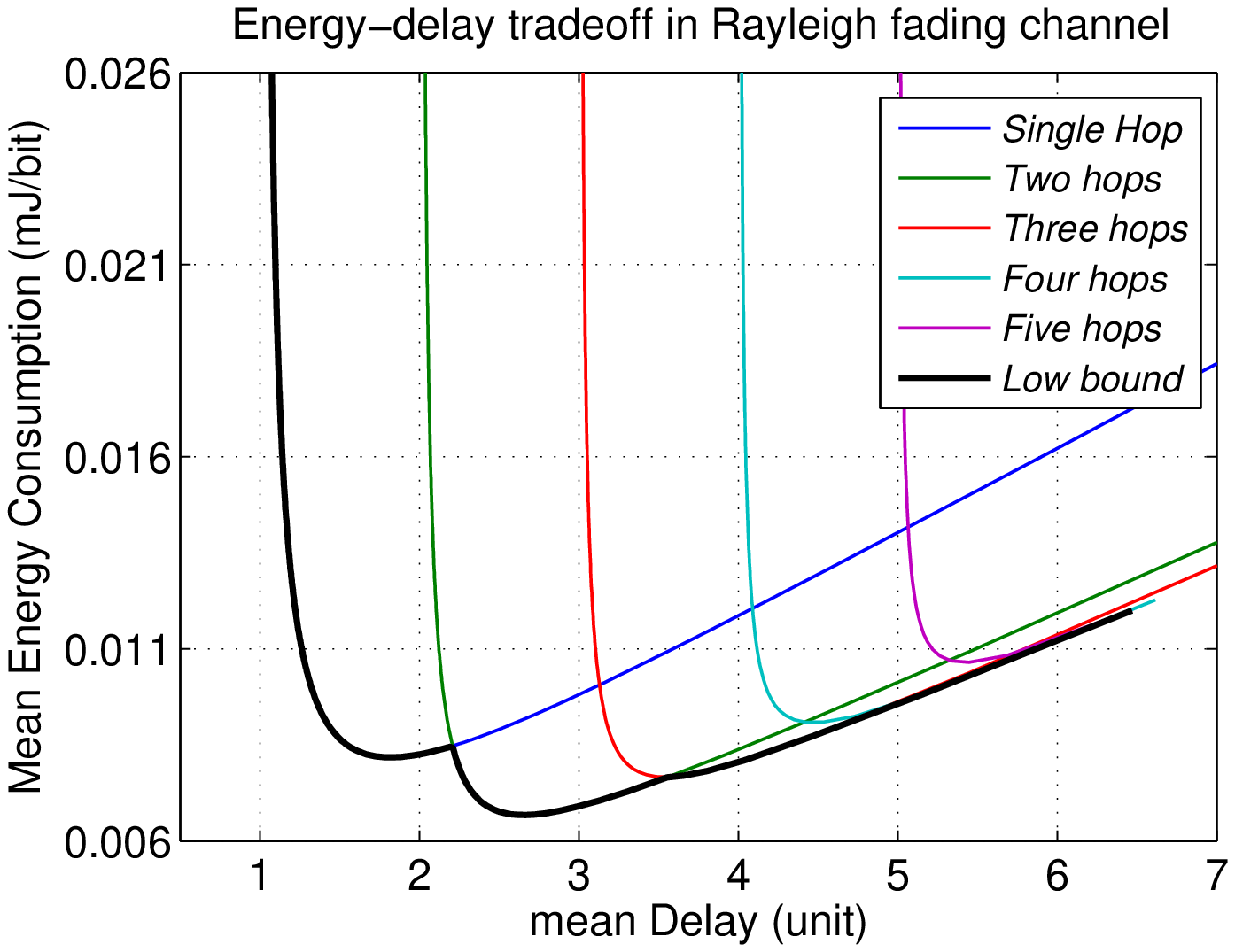}}
    \hfil
    \subfigure[Nakagami block fading channel $d = 380m$]{\includegraphics[width=0.6\textwidth]{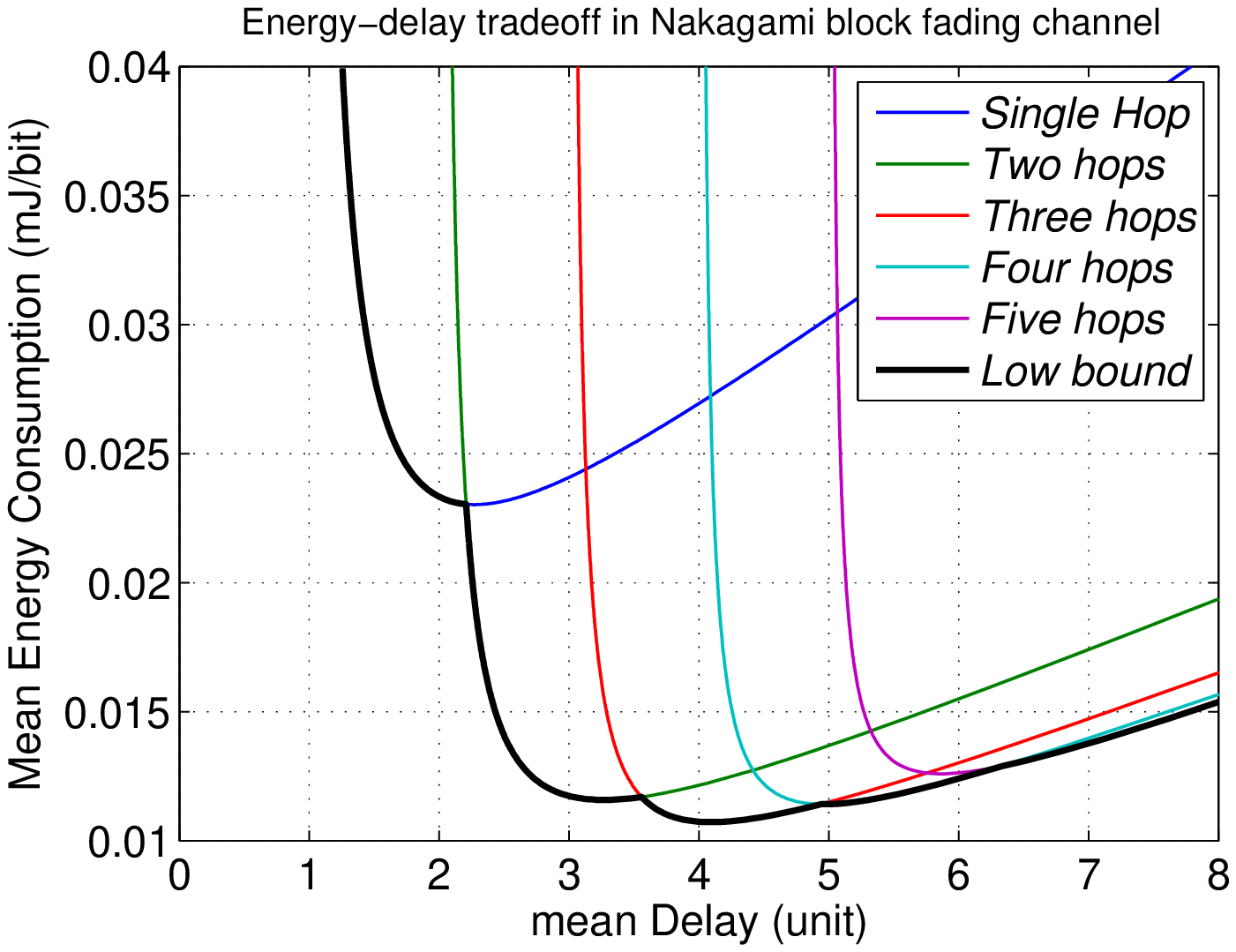}}
    \caption{Energy delay trade-off for the multi-hop transmission}
    \label{fig:ed_to_mh}
\end{figure*}

\begin{theorem} \label{the:mindelay} In
a homogeneous linear network, a source node $x$ sends a packet of
$N_b$ bits to a destination node $x'$ using $n$ hops. The distance
between $x$ and $x'$ is $d$. The length of each hop is $d_1$, $d_2$,
\ldots, $d_n$ respectively and the mean end to end delay is referred
to as $\overline{D}(d)$. The minimum mean end to end delay
$\overline{Dtot}_{min}$ is given by:
\begin{equation}
\overline{Dtot}_{min} =\overline{D}(d/n)\cdot n
\end{equation}
if and only if $d_1=d_2=\ldots=d_n$.
\end{theorem}

\begin{proof} The mean delay of each hop is defined by
$\overline{D}_m, \ m = 1, 2, \ldots, n$. Since each hop is
independent of the other hops, the mean end to end delay is obtained
by:
\[\overline{Dtot} =
\overline{D}_1+\overline{D}_2+\ldots+\overline{D}_n.\] Hence, the
problem can be rewritten as:
\begin{align} &\text{minimize} &&\overline{Dtot}
\notag \\
&\text{subject to}   && d_1+d_2+\ldots +d_n = d.\notag
\end{align}
We set
\[F=\overline{D}_1+\overline{D}_2+\ldots+\overline{D}_n+\lambda(d_1+d_2+\ldots
+d_n-d), \] where $\lambda\neq0$ is the Lagrange multiplier.
According to the method of the Lagrange multipliers, we obtain:
\begin{equation}
\label{eq:lagrange_delay} \left\{
\begin{split}
&\frac{\partial{\overline{D_1}}}{\partial{d_1}}+ \lambda =0\\
&\frac{\partial{\overline{D_2}}}{\partial{d_2}}+ \lambda =0\\
&\ldots\\
&\frac{\partial{\overline{D_n}}}{\partial{d_n}}+ \lambda =0\\
&d_1+d_2+\ldots +d_n = d.
\end{split}
\right.
\end{equation}
Eq.~\eqref{eq:lagrange_delay} shows that the minimum value of $F$ is
obtained in the case
$\frac{\partial{\overline{D_1}}}{\partial{d_1}}=\frac{\partial{\overline{D_2}}}{\partial{d_2}}
=\ldots =\frac{\partial{\overline{D_n}}}{\partial{d_n}} = -\lambda$
Moreover, in a homogeneous network the properties of each node are
identical. Therefore,
\[\frac{\partial{\overline{D_m}}}{\partial{d_m}}=\frac{\partial{\overline{D}}}{\partial{d}}\bigg|_{d=d_m}\]
where $m=1,2,\ldots,n $. Because
$\frac{\partial{\overline{D}}}{\partial{d}}$ is a monotonic
increasing function of $d$ when the path-loss exponent follows
$\alpha\geq2$, the unique solution of Eq.~\eqref{eq:lagrange_delay}
is $d_1=d_2=\ldots=d_n=\frac{d}{n}$. Finally, we obtain:\[
\overline{Dtot}_{min} =  \overline{D}(d/n)\cdot n.
\]
\end{proof}

Based on Theorem~\ref{the:minenergy} and Theorem~\ref{the:mindelay},
we conclude that, regarding a pair of source and destination nodes
with a given number of hops, the only scenario, which minimizes both
mean energy consumption and mean transmission delay, is that each
hop with uniform distance along the linear path.

Fig.~\ref{fig:ed_to_mh} shows the relationship between the mean
energy consumption and the mean delay for a certain transmission
distance in AWGN, Rayleigh and Nakagami block fading channel. The
mean delay is computed with Eq.~\eqref{eq:delay_mh} and the mean
energy consumption is calculated with
Eq.~\eqref{eq:EDRb_d_g},~\eqref{eq:EDRb_d_f} and
\eqref{eq:EDRb_d_b}. According to the Theorems \ref{the:minenergy}
and \ref{the:mindelay}, each relay of the multi-hop transmission
adopts the same transmission power according to the optimal hop
distance. No maximum limit for the transmission power is considered
in the computation. However, it has to be taken into account in
practice.

As shown in Fig.~\ref{fig:ed_to_mh}, we use $d=380m$ for the AWGN
and the Nakagami block fading channel and $d=50$ for the Rayleigh
fading channel. Corresponding optimal number of hops is respectively
$2$ hops, $2$ hops and $3$ hops respectively, which corroborates the
results of Fig.~\ref{fig:ed_to_mh}. The bold black line gives the
mean energy-delay trade-off. Knowing this particular trade-off, the
routing layer can decide how many hops are needed to reach the
destination under a specific transmission delay constraint.

The trade-off curve reveals the relationship between the
transmission power, the transmission delay and the total energy
consumption:
\begin{enumerate}
  \item For smaller delays (fewer hops), more energy is needed due to the high transmission power needed to reach nodes located far away.
  \item An increased energy consumption is not only triggered by communications with few hops but also arises for communications with several hops where the use of a reduced transmission power leads to too many retransmissions, and consequently wastes energy, too. Hence, the decrease of the transmission power does not always guaranty to a reduction of the total energy consumption.
  \item For a given delay constraint, there is an optimal transmission power that minimizes the total energy consumption.
\end{enumerate}

Though the lower bound on the energy-delay trade-off is derived for
linear networks, it will be shown by simulations in the following
section~\ref{sec:sim} that this bound is proper for 2-dimensional
Poisson distributed networks.

\section{Simulations in Poisson distributed networks}\label{sec:sim}
The purpose of this section is to determine the lower bound on the
energy efficiency and on the energy-delay trade-off in a
2-dimensional Poisson distributed network using simulations. The
goal is to show that the theoretical results obtained for a linear
network still hold for such a more realistic scenario. We introduce
this section by defining the characteristic transmission range.

\subsection{Characteristic transmission range}
\label{subsec:dc} The characteristic transmission range is defined
as the range $d_c$ where $EDRb\_1hop(d_c)=EDRb\_2hop(d_c)$, i.e.,
the total energy consumption of a two-hop transmission is equal to
that of a one-hop transmission~\cite{energy:Min02} as shown in
Fig.~\ref{fig:dc}. In a geographical-aware network, the knowledge of
$d_c$ at the routing layer is very useful to decide whether the
optimal transmission can be done in one or two hops. Hence, when the
transmission distance $d$ is greater than $d_c$, the use of a relay
node is beneficial, on the contrary, a direct transmission is more
energy efficient.
\begin{figure}[!t]
\centering
\includegraphics[width=0.8\textwidth]{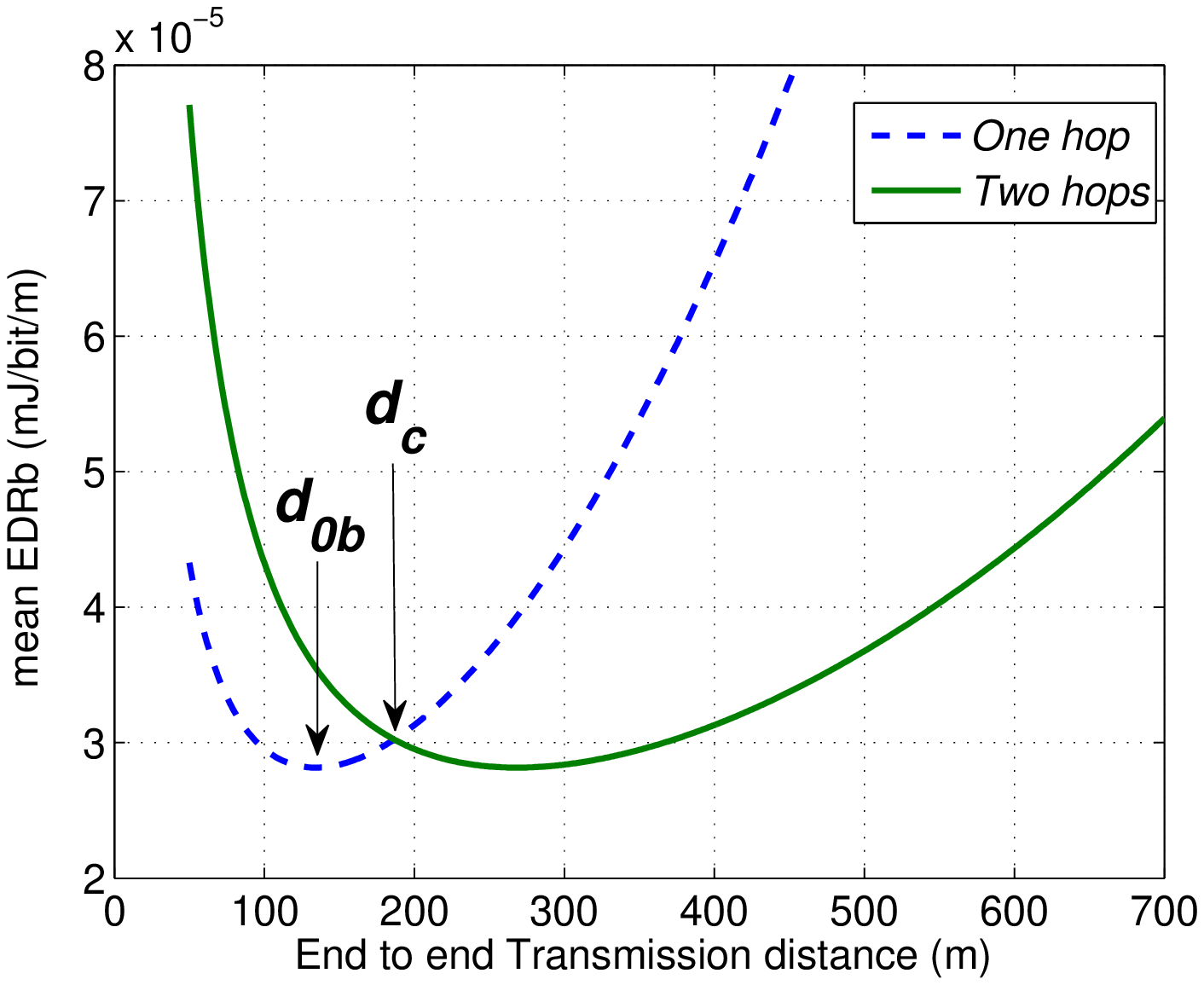}
\caption{Characteristic Transmission Range in a Nakagami block
fading channel where $d_c = 187m$ using the related parameters
listed in Table~\ref{tab:para}} \label{fig:dc}
\end{figure}
\subsection{Simulation setup}
In the simulations, the lower bound on $\overline{EDRb}$ and on the
energy-delay trade-off are evaluated in a square area $A$ of surface
area $S_A=900\times900m^2$. The nodes are uniquely deployed
according to a Poisson distribution:
\begin{equation}
\label{eq:poisson} P(n\: nodes\: in\: S_A)=\frac{(\rho\cdot
S_A)^{n}}{n!}e^{-\rho\cdot S_A}
\end{equation}
where $\rho$ is the node density. All the other simulation
parameters concerning a node are listed in Table~\ref{tab:para}. We
set the node density at $\rho =0.001/m^2$ to ensure a full
connectivity of the network~\cite{con:Gorce07}. The decode and
forward transmission mode is adopted in the simulations.

The network model used in the simulations assumes the following
statements:
\begin{itemize}
    \item The network is geographical-aware, i.e. each node knows the position of
    all the nodes of the network,
    \item A node can adjust its transmission power according to a given transmission range,
    which is determined by the routing layer using Eq.~\eqref{eq:P0_d},~\eqref{eq:P0_d_g}
    or~\eqref{eq:P0_d_b} with respect to AWGN, Rayleigh or Nakagami block fading channel respectively.
    \item A Time Division Multiple Access (TDMA) policy is assumed.
\end{itemize}

\subsection{Simulations of the lower bound on $\overline{EDRb}$}
In these simulations, a very simple routing strategy is adopted as
follows:
\begin{itemize}
\item Step 1: The source node estimates if the distance between
the source and the destination node is smaller than $d_c$; if YES,
transmit packet directly, if NO, go to step 2.
\item Step 2: Select the nodes whose distance from the source are in the
range $d/N_{hop0}\pm(d_c - d_0)$. If no node is chosen, expand the
range step by step (size $d_0$) until reaching the destination node.
\item Step 3: Choose the node closest to the destination node among the nodes chosen in step 2).
\item Step 4: Repeat step 1) to step 3) until the destination node is reached.
\end{itemize}

In the simulations, we test all pairwise source-destination nodes.
The packet size is of $2560$ bits. Then, for each pair of nodes, we
calculate the end to end energy consumption and its Euclidean
distance. The simulation is implemented for $1500$ times, $600$
times and $100$ times respectively corresponding to the node density
$0.0001node/m^2$, $0.0002node/m^2$ and $0.001node/m^2$.

\begin{figure}[!t]
\centering
\includegraphics[width=0.8\textwidth]{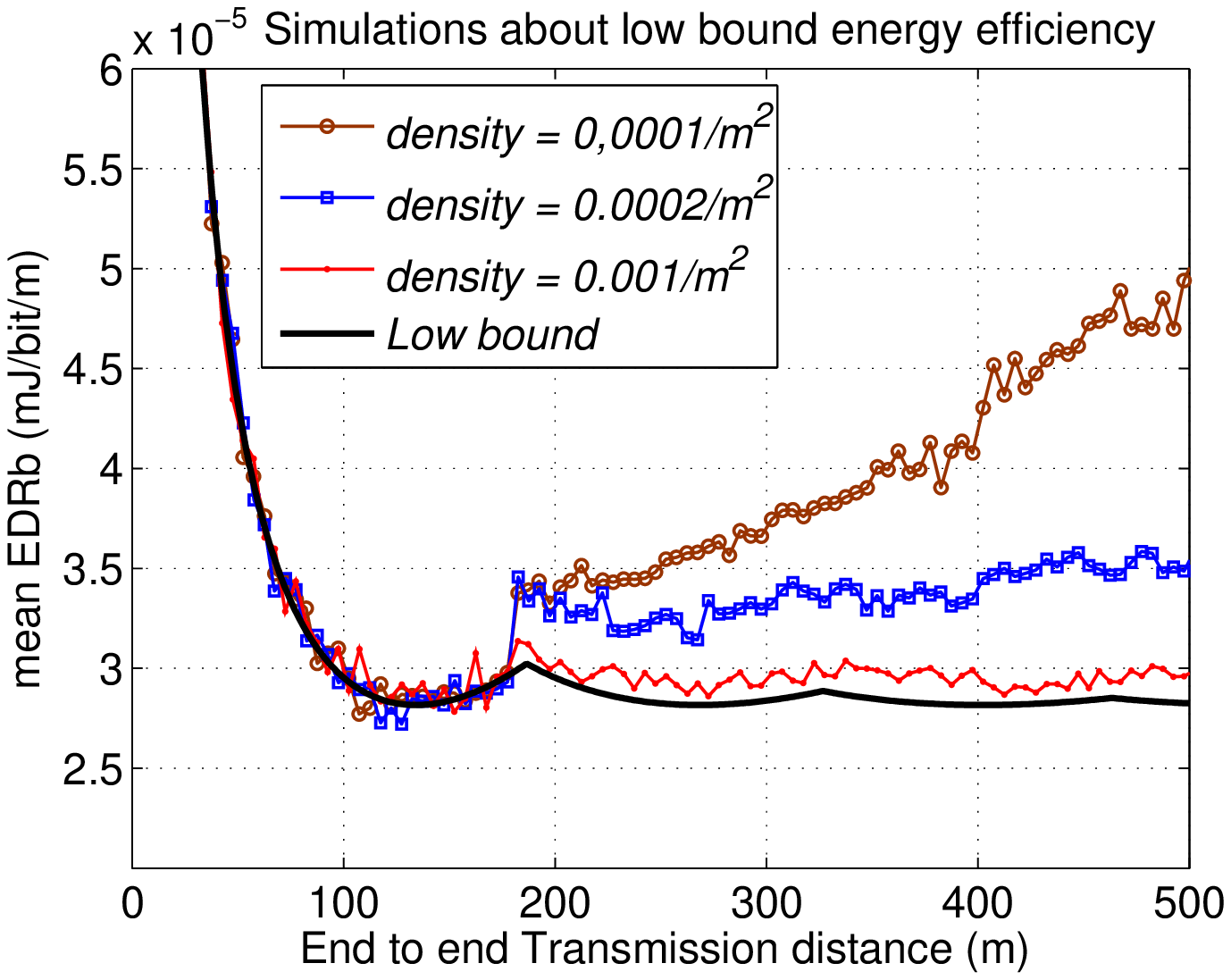}
\caption{Simulation results for the energy efficiency
$\overline{EDRb}$ in Nakagami block fading channel, m
=1.}\label{fig:sim_EDRb_lowbound}
\end{figure}
Fig.~\ref{fig:sim_EDRb_lowbound} shows the simulation results for
the energy efficiency $\overline {EDRb}$ considering different node
densities, a Nakagami block fading channel and a BPSK modulation. We
have $d_0b = 134.16m$ and $d_c = 187m$ in this case. These results
show that:
\begin{enumerate}
  \item The theoretical lower bound on $\overline{EDRb}$ is adequate to
  a 2-D Poisson network although its derivation is based on a linear network.

  When the node density is of $0.001/m^2$, the theoretical lower bound and
  the one obtained by simulations coincide. For this density, a full connectivity
  of the network exists. Hence, we can conclude that our theoretical lower bound
  for the average energy efficiency is suitable for Poisson networks.
  When the node density is reduced, theoretical and simulation based curves for the
  mean $\overline{EDRb}$ diverge when the end to end transmission distance $d$ increases.
  In that case, the source node can not find a relay node in the optimal transmission range
  and has to search for a further relay node which increases the
  energy consumption.

  \item Unreliable links play an important role for energy savings.

  In the simulations, the transmission power is adapted according to the
  transmission distance on the basis of the analysis of
  Section~\ref{sec:one hop}. Hence, unreliable links also contribute to
  attain the lower bound on $\overline{EDRb}$ (as presented in
  Fig.~\ref{fig:EDRb_fading} and~\ref{fig:EDRb_block}, the optimal link probability is about $0.72$).
\end{enumerate}

\begin{figure}[!t]
\centering
\includegraphics[width=0.8\textwidth]{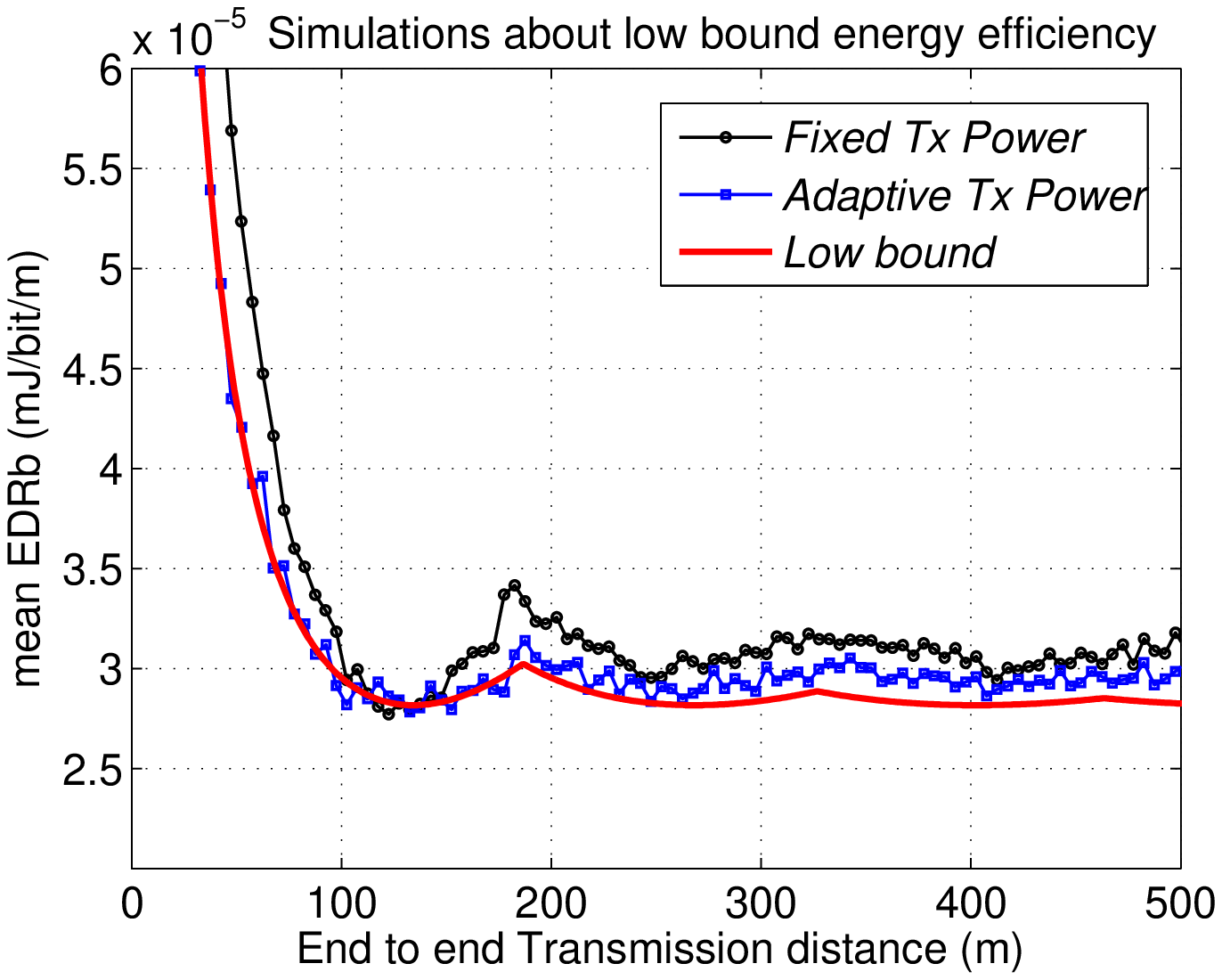}
\caption{Simulation results for $\overline{EDRb}$ in Nakagami block
fading channel using fixed transmission. Here, the node density is
$0.001/m^2$, m=1.}\label{fig:sim_EDRb_fix_pwr}
\end{figure}

Adaptive transmission power is not available in many cheap sensor
nodes. Therefore, we consider a fixed transmission power for each
node in the simulation which is set to the energy-optimal
transmission power of Eq.~\eqref{eq:P0}. Simulation results for a
fully connected network are shown in
Fig.~\ref{fig:sim_EDRb_fix_pwr}. Compared to the adaptive
transmission power mode, nodes with fixed transmission power show a
slightly higher $\overline{EDRb}$, i.e., lower energy efficiency.
Nevertheless, the advantage in terms of simplicity due to the use of
fixed transmission powers makes it worthwhile the little increase in
energy consumption.

\subsection{Simulations of the energy-delay trade-off}
The simulations regarding the energy-delay trade-off are also
implemented for a Nakagami block fading channel and for a fixed
end-to-end transmission distance of $380m$. Regarding each pairs of
 nodes, the source nodes try to use $1$ to $5$ hops in turn.

The following relay selection strategy is adopted knowing the number
of hops:
\begin{itemize}
\item Step 1: Calculate the hop range according to the hop number, i.e., $380m/ hop\ number$.
\item Step 2: Select the set of relay nodes that belong to the 1-hop
transmission range (1-hop length). If the set of relay nodes is
empty, extend the range by 1-hop length until reaching the
destination node.
\item Step 3: Choose the node closest to the destination node among the nodes chosen in step 2).
\item Step 4: Repeat Step 2) and Step 3) until the destination node is reached, then, return selected relay node.
\end{itemize}

The source node and the selected relay node(s) will transmit the
packet with the same transmission power and the value of
transmission power starts form $1dBm$ and increases by $1dBm$ until
$40dBm$. Each simulation is repeated 50 times. Then, we compute the
delay and the energy consumption for each routing. Finally we obtain
the mean delay and mean energy consumption for the same hop number.
In this way, we obtain the low bound of energy-delay trade-off for
three different node densities.

\begin{figure}[!t]
\centering
\includegraphics[width=0.8\textwidth]{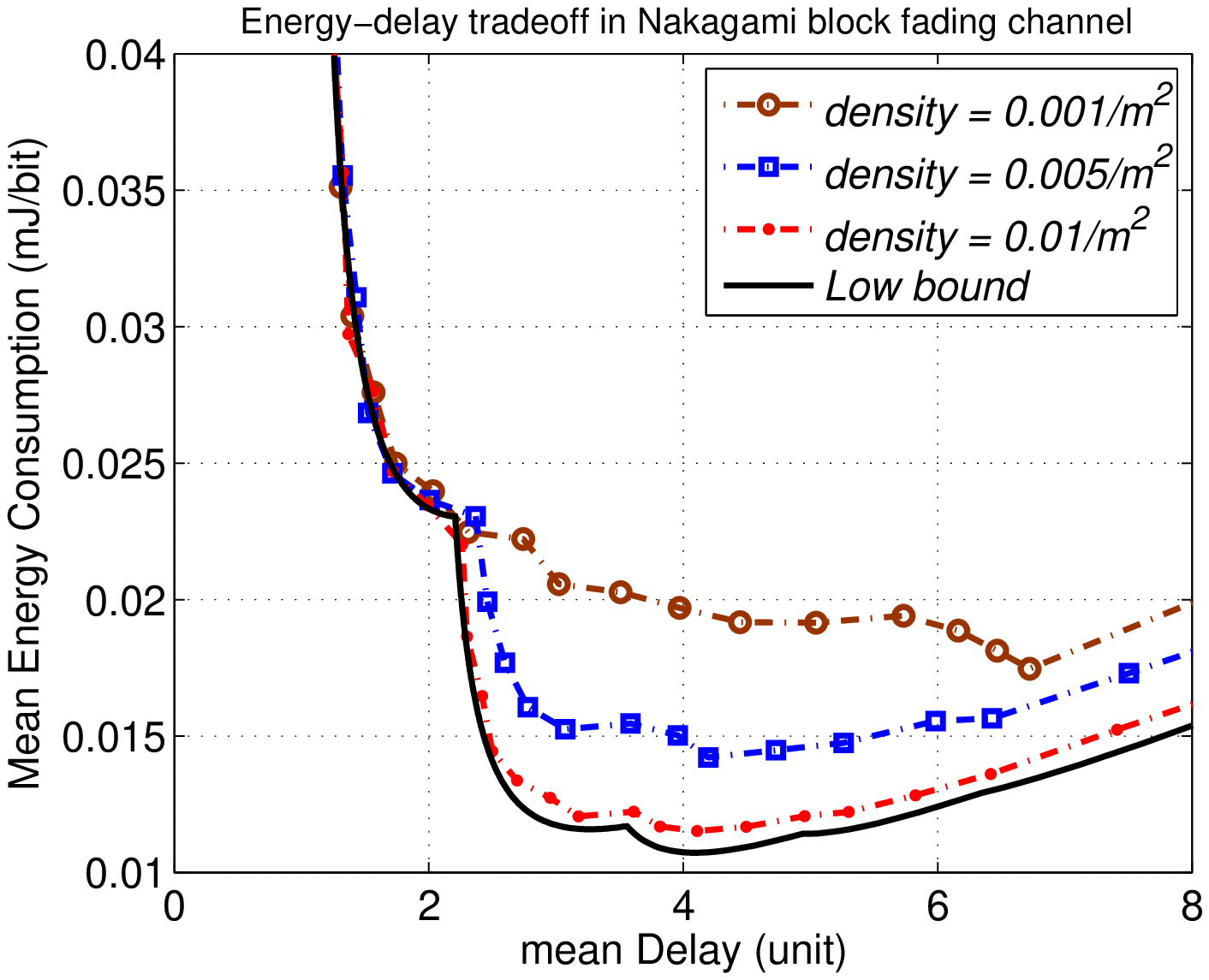}
\caption{Simulation results for the energy-delay trade-off in
Nakagami block fading channel, m =1.} \label{fig:sim_ed_to}
\end{figure}
In Fig.~\ref{fig:sim_ed_to}, simulation results are given for
different node densities. For a node density of $0.01/m^2$, the
lower bound on the energy delay trade-off is reached since there are
enough nodes to find a suitable relay given the delay constraint.
This result indicates that the theoretical lower bound on the
energy-delay trade-off is valid for a Poisson network though its
derivation is based on a linear network. For smaller node densities,
the energy delay trade-off obtained by simulations diverges from the
lower bound since non energy-optimal relays have to be used which
increases the energy consumption and the transmission delay.

\section{Conclusions}
\label{sec:conclu} This paper, using realistic unreliable link
model, explores the low bound of energy-delay trade-off in AWGN
channel, Rayleigh flat fading channel and Nakagami block fading
channel. Firstly, we propose a metric for energy efficiency,
$\overline{EDRb}$, which is combined with the unreliable link model.
It reveals the relation between the energy consumption of a node and
the transmission distance which may contribute to determine optimal
route at the routing layer. By optimizing $\overline{EDRb}$, a
closed form expression of the energy-optimal transmission range is
obtained for AWGN, Rayleigh flat fading and Nakagami block fading
channel. Based on this optimal transmission range, the lower bound
on $\overline{EDRb}$ for a multi-hop transmission using a linear
network is derived for the three different kinds of channel. In
addition, the lower bound on the energy-delay trade-off is studied
for the same multi-hop transmission over a linear network. Results
are then validated using simulations of a 2-D Poisson distributed
network. Theoretical analyses and simulations show that accounting
for unreliable links in the transmission contributes to improve the
energy efficiency of the system under  delay constraints, especially
for Rayleigh flat fading and Nakagami block fading channel.
\appendix
\section{Derivation of $P_{0}$ and $d_0$}
\label{sec:appendix} Substituting \eqref{eq:pl} and \eqref{eq:gamma}
into \eqref{eq:dPt_EDRb} and \eqref{eq:dd_EDRb}, we obtain:

\begin{equation}
\label{eq:proof_P0_1} \left\{\begin{split}
&\frac{d^\alpha K_1 p_l(\gamma)-K_2(E_c+K_1 P_t)p_l^\prime(\gamma)}{d^{(\alpha+1)}p_l^2(\gamma)}&=0\\
&\frac{(E_c+K_1 P_t)\cdot(-d^\alpha p_l(\gamma)+K_2 P_t \alpha
p_l^\prime(\gamma))}{d^{(\alpha+2)}p_l^2(\gamma)} &=0
\end{split}\right.
\end{equation}
where $p_l^\prime(\gamma)$ is the derivative of the function
$p_l(\gamma)$. Because $E_c+K_1 P_t$ is greater than $0$,
simplifying the equation set~\eqref{eq:proof_P0_1} yields:
\begin{equation}
\label{eq:proof_P0_2} \left\{\begin{split}
&d^\alpha K_1 p_l(\gamma)-K_2(E_c+K_1 P_t)p_l^\prime(\gamma)&=0\\
&-d^\alpha p_l(\gamma)+K_2 P_t \alpha p_l^\prime(\gamma) &=0
\end{split}\right.
\end{equation}
Solving the equation set~\eqref{eq:proof_P0_2} and substituting
$\gamma$ with \eqref{eq:gamma}, we have:
\begin{equation} P_0 = \frac{E_c}{K_1(\alpha -1)} \notag
\end{equation}
and
\begin{equation}
d_0 ^\alpha = \frac{p_l^\prime (K_2 P_0 d_0^{-\alpha}) K_2
P_0}{p_l(K_2 P_0 d_0 ^{-\alpha}) }. \notag
\end{equation}

\section{Derivation of the optimal transmission range in AWGN
channel} \label{sec:d0g_der_awgn}According to \eqref{eq:pl}, the
link model in AWGN channel is given by:
\begin{equation}
\label{eq:pl_gaussian} p_l = (1 - BER(\gamma))^{N_b},
\end{equation}
where $BER(\gamma)$ is the Bit Error Rate (BER). A closed form of
BER is described in~\cite{com:Goldsmith05} for coherent detection in
AWGN channel:
\begin{equation}
\label{eq:ber_awgn_general} BER(\gamma_b) = \alpha_m
{\textrm{Q}}(\sqrt{\beta_m \gamma_b}),
\end{equation}
with the Q function, $\textrm{Q}(x)=
\int_{x}^{\infty}\frac{1}{\sqrt{\pi}} e^{-u^2/2} du$, where
$\alpha_m$ and $\beta_m$ rely on the modulation type and order,
e.g., for Multiple Quadrature Amplitude Modulation (MQAM)
$\alpha_m=4(1-1/\sqrt{M})/\log_2(M)$ and $\beta_m=3\log_2 (M)/(M-1)$
. For Binary Phase Shift Keying (BPSK), $\alpha_m=1$ and $\beta_m =
2$.

The closed form expression of $d_0$ can not be obtained using the
exact $BER(\gamma)$. A simplified tight approximation of
$BER(\gamma)$ is obtained when $\beta_m\cdot \gamma_b\geq2$ by using
the method proposed in~\cite{math:Ermolova04}:
\begin{equation}
\label{eq:ber_awgn} BER_g(\gamma_b)\approx 0.1826\alpha_m\cdot
\exp(-0.5415\beta_m\gamma_b) \ \ \ \text if\ \beta_m\cdot
\gamma_b\geq2,
\end{equation}
where $\exp(\cdot)$ represents the exponential function.
Fig.~\ref{fig:BER_awgn_approx} shows the relation between the
approximation and the exact values of the BER.
\begin{figure}[!t]
\centering
\includegraphics[width=0.8\textwidth]{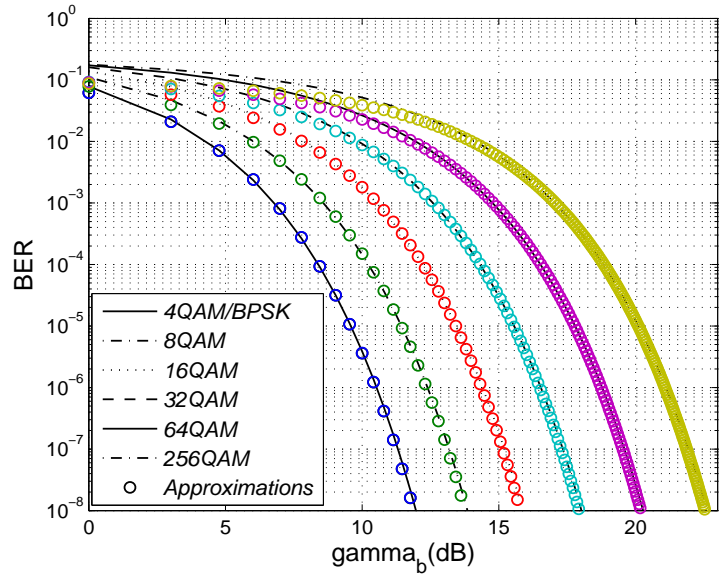}
\caption{BER approximations for MQAM} \label{fig:BER_awgn_approx}
\end{figure}

Therefore, the optimal transmission range $d_{0g}$ is obtained by
substituting~\eqref{eq:ber_awgn} and~\eqref{eq:P0}
into~\eqref{eq:dd_EDRb}:
\begin{equation}
d_{0g} = \left(\frac{-0.5415\beta_m K_2 N_b E_c \alpha}{K_1(\alpha
-1)(1+\alpha N_b \text
W_{-1}\left[\frac{-e^{-\frac{1}{N_b\cdot\alpha}}}{0.1826\alpha_m N_b
\alpha}\right])}\right)^{\frac{1}{\alpha}},
\end{equation}
where $\text W_{-1}[\cdot]$ is the branch satisfying $\text W(x)<-1$
of the Lambert W function~\cite{math:Corless96}.

\section{Derivation of the optimal transmission range in Rayleigh flat
fading channel} \label{sec:d0f_der_fading} There is a general
expression for the BER in Rayleigh flat fading channel in case of
$\bar\gamma\geq5$ in~\cite{com:Goldsmith05}:
\begin{equation}
\label{eq:berf} BER_{f}(\bar\gamma) \approx \frac{\alpha_m}{2\beta_m
\bar\gamma}
\end{equation}
where $\alpha_m$ and $\beta_m$ are the same as those
in~\eqref{eq:ber_awgn}.

Substituting~\eqref{eq:berf} and~\eqref{eq:gamma}
into~\eqref{eq:EDRb}, we have:
\begin{equation}
\label{eq:EDRb_fading} \overline{EDRb_f}(d,P_t)=
\frac{E_{c}+K_1\cdot P_{t}}{d\left(1-\frac{\alpha_m}{2\beta_m
K_2\cdot P_t \cdot d^{-\alpha}}\right)^{N_b}}
\end{equation}

Substituting~\eqref{eq:EDRb_fading} and~\eqref{eq:P0}
into~\eqref{eq:dPt_EDRb}, the optimal transmission range $ d_{0f}$
in Rayleigh flat fading channel is obtained:
\begin{equation}
d_{0f}=\left(\frac{2\beta_m E_c\cdot{K_2}}{(\alpha -1)
\cdot{K_1}\alpha_m (\alpha N_b+1)}\right)^{\frac{1}{\alpha}}
\end{equation}

\section{Derivation of the optimal transmission range in Nakagami block
fading channel} \label{sec:d0b_der_block} The exact link model in
Nakagami block fading channel is~\cite{con:Gorce07}:
\begin{equation}
p_l(\bar{\gamma}) = \int_{\gamma=0}^{\infty} (1-BER(\gamma))^{N_b}
p(\gamma | \bar\gamma)d\gamma, \label{eq:pl_block}\end{equation}
where
\begin{equation}
p(\gamma | \bar\gamma) = \frac{m^m
\gamma^{m-1}}{\bar\gamma\Gamma(m)}\exp\left(-\frac{m\gamma}{\bar\gamma}\right),
\end{equation}
$BER(\gamma)$ refers to~\eqref{eq:ber_awgn}.

When $m =1$ (Rayleigh block fading) and $\alpha_m =1$, the
approximation of \eqref{eq:pl_block} is found:
\begin{equation}
\label{eq:pl_block_approx} pl(\overline{\gamma})=
\exp\left(\frac{-4.25\log_{10}(Nb)+2.2}{\beta_m
\overline{\gamma}}\right).
\end{equation}
\begin{figure}[!t]
\centering
\includegraphics[width=0.8\textwidth]{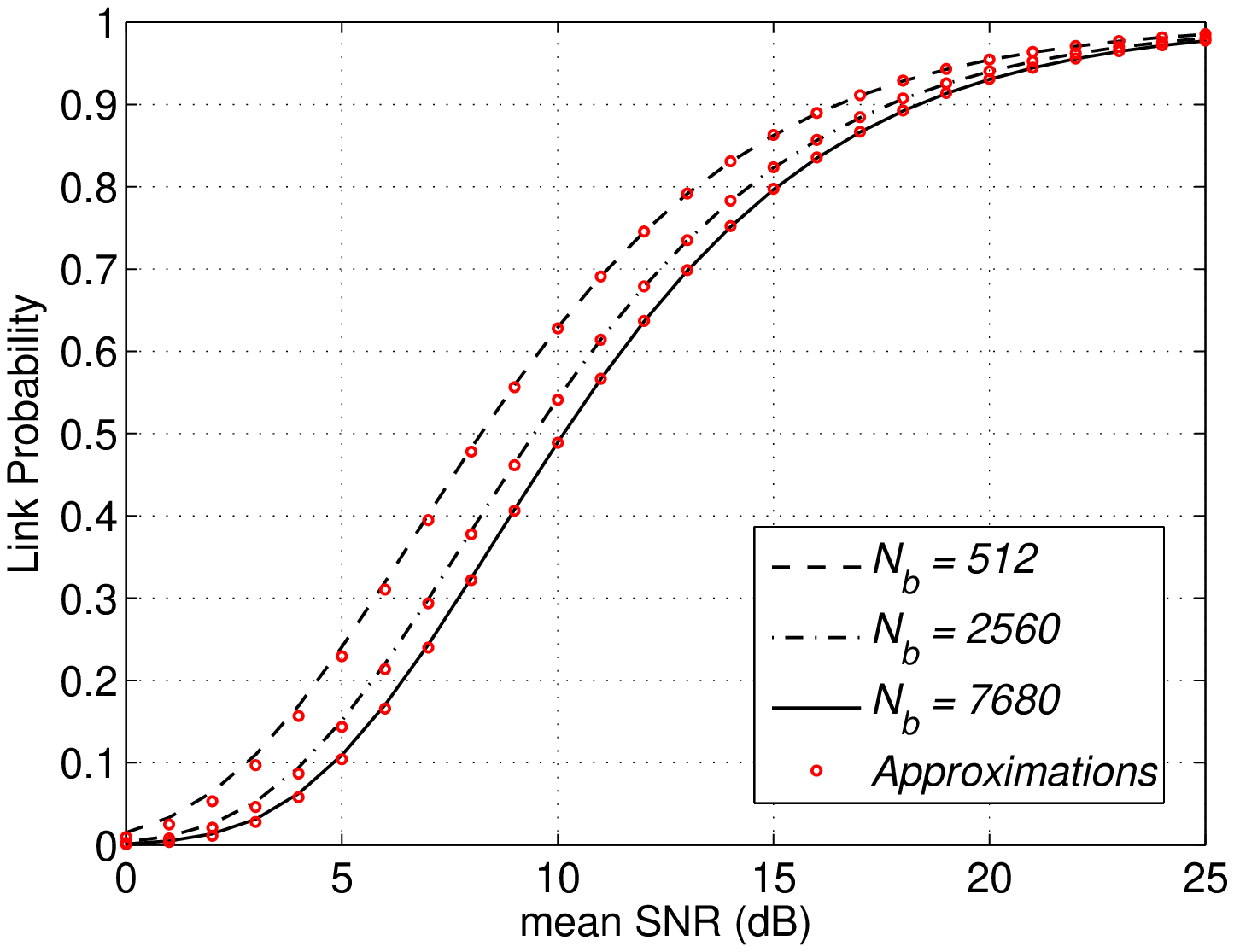}
\caption{Approximations of link probability for Rayleigh block
fading} \label{fig:approx_link_block}
\end{figure}
Fig.~\ref{fig:approx_link_block} shows the approximations for
different values of $N_b$.

Substituting~\eqref{eq:pl_block_approx} and \eqref{eq:P0} into
\eqref{eq:dPt_EDRb} yields the optimal transmission range $d_{0b}$:
\begin{equation}
d_{0b}=\left(\frac{ \beta_m K_2 E_c
}{K_1(\alpha^2-\alpha)(4.25\log_{10}(N_b)-2.2)} \right) ^{1/\alpha}.
\end{equation}
\newpage
\tableofcontents
\newpage
\bibliographystyle{ieeetr}
\bibliography{mybib,books}
\end{document}